\documentclass{article}
\usepackage{fullpage}
\usepackage[T1]{fontenc}
\usepackage[utf8]{inputenc}
\usepackage[colorlinks=true,allcolors=blue]{hyperref}
\usepackage{amsmath, amssymb, amsthm, amsfonts, graphicx,color,subcaption, enumerate, bm, cleveref, enumitem, array}

\usepackage{multicol, multirow}
\usepackage{thmtools, thm-restate}
\usepackage{ragged2e}
\usepackage{lineno}
\usepackage{subfiles}

\newtheorem{theorem}{Theorem}[section]
\newtheorem{lemma}[theorem]{Lemma}
\newtheorem{corollary}[theorem]{Corollary}
\newtheorem{problem}[theorem]{Problem}
\newtheorem{claim}[theorem]{Claim}
\newtheorem{observation}[theorem]{Observation}
\newtheorem{definition}[theorem]{Definition}

\newcommand{\OO}{\widetilde{O}}
\newcommand{\Ot}{\OO}
\newcommand{\eps}{\varepsilon}
\newcommand{\R}{\mathbb{R}}
\newcommand{\cF}{\mathcal{F}}
\newcommand{\E}{\mathbb{E}}
\renewcommand{\emptyset}{\varnothing}

\newcommand{\dirb}{\vec{b}}
\newcommand{\dirB}{\overrightarrow{B}}
\newcommand{\dirBs}{\overrightarrow{\mathcal{B}}}
\newcommand{\dirX}{\overrightarrow{X}}
\newcommand{\dirLP}{\overrightarrow{\mathcal{LP}}}
\newcommand{\dirMB}{\overrightarrow{MB}}
\newcommand{\dirMBs}{\overrightarrow{\mathcal{MB}}}

\newcommand{\Ghat}{\widehat{G}}

\DeclareMathOperator{\Pdim}{\mathrm{Pdim}}
\DeclareMathOperator{\VCdim}{\mathrm{VCdim}}
\DeclareMathOperator{\Bunch}{\mathsf{Bunch}}

\newcommand{\bnd}{\partial}
\newcommand{\rdiv}{\mathcal{R}}
\newcommand{\dist}{d}

\def\poly{\operatorname{poly}}
\def\polylog{\operatorname{polylog}}

\newcommand{\email}[1]{\href{mailto:#1}{#1}}
\title{Subquadratic algorithms in minor-free digraphs: (weighted) distance oracles, decremental reachability, and more}
\author{Adam Karczmarz\thanks{University of Warsaw and IDEAS NCBR, Poland. \email{a.karczmarz@mimuw.edu.pl}. Partially supported by the ERC
CoG grant TUgbOAT no 772346 and the National Science Centre (NCN) grant no. 2022/47/D/ST6/02184.} \and Da Wei Zheng\thanks{University of Illinois Urbana-Champaign. \email{dwzheng2@illinois.edu}.}}
\date{}

\begin{document}

\maketitle

\begin{abstract}
Le and Wulff-Nilsen [SODA '24]  initiated a systematic study of VC set systems to unweighted $K_h$-minor-free directed graphs. We extend their results in the following ways:
\begin{itemize}
    \item We present the first application of VC set systems for \emph{real-weighted} minor-free digraphs to build the first \emph{exact} subquadratic-space distance oracle with $O(\log n)$ query time.
    All prior work using VC set systems only applied in \emph{unweighted} graphs and digraphs.
    \item We describe a unified system for analyzing the VC dimension of balls and the LP set system (based on Li--Parter [STOC '19]) of Le--Wulff-Nilsen [SODA '24] using \emph{pseudodimension}.
    This is a major conceptual contribution that allows for both improving our understanding of set systems in digraphs as well as improving the bound of the LP set system in directed graphs to $h-1$.
    \item We present the first application of these set systems in a \emph{dynamic} setting. Specifically, we construct decremental reachability oracles with subquadratic total update time and constant query time.
    Prior to this work, it was not known if this was possible to construct oracles with subquadratic total update time and polylogarithmic query time, even in planar digraphs.
    \item We describe subquadratic time algorithms for unweighted digraphs including (1) constructions of exact distance oracles, (2) computation of vertex eccentricities and Wiener index.
    The main innovation in obtaining these results is the use of dynamic string data structures.
\end{itemize}
\end{abstract}
\thispagestyle{empty}
\clearpage
\pagenumbering{arabic}

\section{Introduction}
Computing all-pairs shortest paths (APSP) in a weighted graph $G$, and related graph parameters concerning distances, such as the girth, diameter, radius, and Wiener index, are among the most fundamental and well-studied graph problems.
It is conjectured that computing APSP requires $n^{3-o(1)}$ time~\cite{WilliamsW18}, or more specifically $mn^{1-o(1)}$ time~\cite{LincolnWW18}, where $n$ and $m$ denote the numbers of vertices and edges of $G$, respectively.
These conditional lower bounds imply the best-known algorithms for APSP~\cite{Pettie04, Williams18} are near-optimal.

Surprisingly, computing most of the aformentioned single-number graph parameters is as hard as solving APSP~\cite{LincolnWW18, WilliamsW18}.
Computing the diameter, whose relationship to APSP remains unclear, is an exception here.
Still, computing it requires $n^{2-o(1)}$ time unless the Strong Exponential Hypothesis fails~\cite{RodittyW13}. This lower bound matches the APSP upper bound for sparse graphs up to subpolynomial factors. 

Improved, truly subcubic algorithms for APSP and related problems can be obtained for \emph{unweighted} or small-weighted dense graphs via algebraic techniques, e.g.,~\cite{CyganGS15, RodittyW11, Seidel95, Zwick02}.
However, for sparse graphs, we lack non-trivial algorithms even in the unweighted case, and even if the graph is undirected. While for APSP such achievements are impossible since the output has quadratic size, for graph parameters this is not an issue.
Nevertheless, truly subquadratic algorithms for the aforementioned graph parameters are ruled out in the unweighted undirected case as well by believable conjectures~\cite{AbboudWW16, LincolnWW18, RodittyW13}.

Since APSP is uninteresting in sparse graphs, one often considers a more flexible variant, the \emph{distance oracle} problem. A distance oracle is a data structure answering distance queries for pairs of vertices in~$G$.
When designing distance oracles one would like to optimize the space and the query time, although keeping the preprocessing time small is also desirable for applications.
That being said, non-trivial (i.e., with sublinear query time) exact distance oracles with subquadratic preprocessing are unlikely to exist for sparse graphs (by, e.g.,~\cite{PatrascuR14,DalirrooyfardJW22}) and whether non-trivial subquadratic-space oracles are possible for sparse graphs remains a very interesting open problem.

\paragraph{Graph classes.} The impossibility results in sparse graphs
motivate the search for algorithms exploiting structural properties of special graph classes, especially if one is interested in \emph{exact} answers.
For instance, distance oracles and computing distance-related graph parameters have been studied extensively in planar (directed) graphs, e.g.,~\cite{FakcharoenpholR06, MozesNNW18, MozesS12}. The breakthrough application of Voronoi diagrams~\cite{Cabello19} led to subquadratic algorithms for diameter, radius, Wiener index~\cite{Cabello19, GawrychowskiKMS21}, and exact distance oracles with $n^{1+o(1)}$ space and $O(\polylog{n})$ query time, and $\Ot(n^{3/2})$ preprocessing time~\cite{CharalampopoulosGLMPWW23}. All these results assume the most general setting of real-weighted directed planar graphs.

The most powerful techniques that exploit the planarity of the graphs for distance problems, like the Monge property (used e.g. in~\cite{FakcharoenpholR06}), or Voronoi diagrams, are sometimes applicable to higher genus graphs but seem completely hopeless beyond that, e.g., for apex graphs or, even more generally, $K_h$-minor-free graphs (which are sparse for constant $h$). 

In this paper, we study distance oracles and related graph parameters computation in $K_h$-minor-free \emph{directed} graphs, for $h=O(1)$.
 
\paragraph{$K_h$-minor-free digraphs and VC set systems.} One important structural property of planar graphs that seamlessly transfers to minor-free graphs is the existence
of efficiently computable $O(\sqrt{n})$-sized balanced separators.
This is enough for obtaining subquadratic solutions for some of the aforementioned problems, such as computing the girth or constructing a subquadratic-space distance oracle with sublinear query time~\cite{djidjev1996efficient}.
However, balanced separators alone seem not powerful enough to break through the quadratic barrier for other graph parameters like diameter, which, in the case of planar graphs, was first enabled by Voronoi diagrams \cite{Cabello19}.

Not long after the Voronoi diagrams breakthrough for planar graphs, a new promising algorithmic approach to distance problems emerged: \emph{VC set systems}.
First applied by Li and Parter~\cite{LiP19},
analyzing the number of possible \emph{distance patterns} via bounding their VC dimension provided an alternative way to compute the diameter of a planar graph.
While the approach worked only in the \emph{unweighted undirected} case and did not yield as efficient algorithms as the Voronoi diagrams machinery, it was robust and simple enough to be applicable e.g. in the distributed setting~\cite{LiP19}.
This robustness yields applicability of the approach to more general graph classes. Using different set systems~\cite{ChepoiEV07}, \cite{DucoffeHV22} showed that diameter and radius (or, more generally, the $n$ vertex eccentricities%
\footnote{The eccentricity of a vertex is the maximum distance from that vertex to another.
Diameter is the maximum vertex eccentricity, whereas radius is the minimum one.})
can be computed in subquadratic time in undirected unweighted graphs with both sublinear balanced separators and constant distance VC dimension, a class that includes $K_h$-minor-free graphs.
Recently, Le and Wulff-Nilsen~\cite{LeW24} explored the approach further specifically for $K_h$-minor-free graphs, and gave improved algorithms (compared to~\cite{DucoffeHV22}) for computing the distance-related graph parameters. In particular, they gave the first subquadratic algorithm for the Wiener index in minor-free graphs.

More importantly,~\cite{LeW24} applied the VC set systems approach to \emph{directed} minor-free graphs for the first time.
They did so by analyzing and bounding the VC dimension of set systems used in~\cite{DucoffeHV22}~and~\cite{LiP19} in the directed case, also giving a more general variant of the latter. 
Building on that, in unweighted $K_h$-minor-free directed graphs they obtained first subquadratic algorithms for computing vertex eccentricities. They also described a distance oracle with subquadratic space and $O(\log{n})$ query time. 

\subsection{An algorithmic template}\label{sec:template}

Before we discuss our contributions, it is useful to first explain a high-level overview of how VC set systems have been used to design subquadratic algorithms for minor-free graphs.
(In this section, to serve both as a warmup and to make our definitions concrete, we give a sketch of how to construct an $O(1)$-query time reachability oracle in a $K_h$-minor-free digraph $G$ as a running example.\footnote{To the best of our knowledge, this easy construction has not been described prior to our work. See also~\Cref{sec:approx} for a description of a more general $(1+\eps)$-approximate distance oracle for unweighted digraphs that is very similar.} Note that reachability can be viewed as a finite approximation to distances.)
We begin by defining the two main tools involved: VC set systems and $r$-divisions.

\paragraph{Set systems.} A set system is formally a collection $\mathcal{F}$ of subsets of a ground set $\mathcal{U}$. A well-known Sauer-Shelah lemma states that if the VC dimension of $\mathcal{F}$ is bounded by $d$, then $|\mathcal{F}|=O\left(|\mathcal{U}|^d\right)$. 
Importantly, the VC dimension does not increase when \emph{restricting} the ground set to its subset ${\mathcal{U}'\subset \mathcal{U}}$. As a result, the upper bound $d$ on the VC dimension of $\mathcal{F}$ also implies $|\mathcal{F}'|=O\left(|\mathcal{U}'|^d\right)$, where \linebreak
$\mathcal{F}':=\{S\cap \mathcal{U}':S\in \mathcal{F}\}$ is called the restriction of $\mathcal{F}$ to $\mathcal{U}'$.

A set system $\mathcal{S}$ used in distance-related applications typically uses $V$, or $V$ annotated with some auxiliary data, as the ground set. For example, the set system $\dirBs_G\subseteq 2^V$ used in~\cite{DucoffeHV22, LeW24} consists of \emph{directed balls} in $G$, ranging over all possible centers in $V$ and radii in $\mathbb{R}_{\geq 0}$. 
(For constructing a reachability oracle, we will use the set of balls centered at every vertex of $V$ with infinite radius.)

\paragraph{$r$-divisions.} The second crucial ingredient is the so-called $r$-division~\cite{Frederickson87} of $G$. For any $r\in [1,n]$, an $r$\nobreakdash-division $\rdiv$ of $G$ is a collection of $O(n/r)$ edge-induced subgraphs of~$G$, called \emph{pieces}, satisfying the following. First, the union of the pieces in $\rdiv$ equals~$G$. Second, each piece $P\in\rdiv$ has $O(r)$ edges.
Moreover, the \emph{boundary} $\bnd{P}$ of a piece $P$, defined as the subset of vertices of $P$ shared with some other piece in $\rdiv$, has size $O(\sqrt{r})$.
The boundary $\bnd{\rdiv}$ of $\rdiv$ is defined as the union of the individual pieces' boundaries.
An $r$-division of a $K_h$-minor-free graph (for $h=O(1)$) exists for any~$r$ and can be computed in $O(n^{1+\eps})$ time, by plugging in the balanced separator algorithm of~\cite{KawarabayashiR10} into the partitioning algorithm of~\cite{Frederickson87}.

\paragraph{Combining the two.} Let $\rdiv$ be an $r$-division of $G$, and let $\mathcal{S}$ be a set system of VC dimension $O(1)$. 
Roughly speaking, we define the set of \emph{patterns} of a piece $P\in\rdiv$ to be the restriction of $\mathcal{S}$ to either $V(P)$ or $\bnd{P}$, whichever better suits the application.
The first basic objective when choosing the system is that the number of patterns per piece is $\poly(r)$, and restricting a set system with constant VC dimension to the piece's (boundary) vertices ensures this.
(For our running example of a reachability oracle, we can choose the set of all balls $\dirBs_G$ which has VC dimension $h-1$. 
Each pattern in this set system is the set of vertices in $\bnd{P}$ reachable from some vertex $v\in G$.
By the Sauer-Shelah lemma, the number of patterns for this set system is $O\left(|\bnd{P}|^{h-1}\right) = O\left(r^{(h-1)/2}\right)$.)

Additionally, for any piece $P\in\rdiv$ and vertex $s\notin V(P)$, one should be able to map $(s,P)$ to one or more (but surely $o(r)$) patterns $\delta$ of the piece $P$, so that some desired information about the distances from $s$ to $V(P)$ can be decoded in $o(r)$ time based exclusively on:
\begin{enumerate}[label=(\arabic*)]
    \item ``global'' properties from $s$ to $\bnd{P}$ in $G$ like reachability or distances,
    \item ``local'' data associated with the pattern(s) $\delta$ that is computable in polynomial time for each such pattern based on $\delta$ and $P$ only and once. (In our running example, we can let $\delta$ be the set of reachable nodes from $s$ on the boundary $\bnd{P}$. This uniquely determines the set of reachable nodes from $s$ within~$P$.)
\end{enumerate}
The information of the former kind can be computed and stored in $\Ot(n^2/\sqrt{r})$ time and space through all pairs $(s,P)$.
Computing and storing the local data takes $O(n\poly(r))$ time and space. 
Therefore, the total space required  is $O(n^2/\sqrt{r} + n\poly(r))$.
For a sufficiently small-polynomial value of $r$, 
handling both the local and the global data takes $n^{2-\Theta(1)}$ time.
(For reachability, the pattern's local data is computed in $O(r)$ time, so the oracle's construction takes $O(n^2/\sqrt{r}+nr^{(h-1)/2})$ time. Setting $r=n^{2/h}$ yields $O(n^{2-1/h})$ time.)

Given the local and global data above, another challenge lies in the subquadratic computation of patterns that the $O(n^2/r)$ pairs $(s,P)$ map to. 
(For computing reachable set of vertices this is trivial to do in $O(n^2/\sqrt{r})$ time.)
Finally, extracting the information from local and global data for all pairs $(s,P)$ should also take subquadratic time. (Given a set of vertices $\delta$ on the boundary of a piece $\bnd P$ that are reachable from a vertex~$v$, deciding whether $w\in V(P)$ is reachable from $v$ in $G$ amounts to testing the membership $w\in \delta$. This takes $O(1)$ time, since $\delta$ is explicitly stored.)

\subsection{Our contribution}
In this paper, we further investigate the capabilities of VC set systems (and generalizations of these set systems) for solving distance-related problems in \emph{$K_h$-minor-free digraphs}. We prove results in weighted, unweighted and dynamic settings.

\paragraph*{A weighted distance oracle.}
First of all, for the first time, we show the applicability of the technique to an exact distance-related problem with real edge weights (but no negative cycles). Specifically, we show:
\begin{theorem}\label{thm:weighted_do}
    Let $G$ be a real-weighted $K_h$-minor digraph. There exists an exact distance oracle for $G$ using $O\left(n^{2-1/(4h-1)}\right)$ space and supporting arbitrary-pair distance queries in $O(\log{n})$ time.
\end{theorem}
To the best of our knowledge, all the previous applications of VC set systems~\cite{DucoffeHV22, LeW24, LiP19} were inherently tailored to unweighted graphs%
\footnote{\cite{DucoffeHV22} extended their results on computing diameter to integer weighted graphs, but had dependence in $\log M$ where $M$ was the size of the largest weight. Their algorithm could only give $(1+\eps)$-approximations for real-weighted graphs.}. 
In particular, the comparable distance oracles for $K_h$-minor-free digraphs~\cite{LeW24} crucially relied on the fact that $G$ is unweighted.

One can obtain a distance oracle for weighted $K_h$-minor-free graphs with $O(n^2/\sqrt{r})$ space and $O(\sqrt{r})$ query time for any $r\in [1,n]$ using $r$-divisions alone~\cite{djidjev1996efficient}, but this does not achieve polylogarithmic query times while using subquadratic space, or even a subquadratic space-query time product. Such a tradeoff (in fact an almost optimal one) is only known for planar graphs ~\cite{CharalampopoulosGLMPWW23}.

\paragraph*{A unified system with an improved pseudodimension bound.}
\cite{LeW24} analyzed two distance VC set systems for directed minor-free graphs: the aforementioned $\dirBs_G$ (for ``balls'', as used in~\cite{DucoffeHV22}) and $\dirLP_{G,M}$ (for Li-Parter, inspired by~\cite{LiP19}); see~\Cref{sec:setsystems} for details. They proved that their VC dimensions
are bounded by $h-1$ and~$h^2$, respectively; the arguments used for the two bounds differed.

\cite{LeW24} used patterns based on the $\dirBs_G$ system to obtain a distance oracle for unweighted graphs with space $O(n^{2-1/(2h-1)})$ and $O(\log{n})$ query time\footnote{\cite[Section~4.3.1]{LeW24} miscalculated the space consumption of their directed oracle to be $O(nr^{h-3/2}+n^2/\sqrt{r})$, which led the bound $O(n^{2-1/(2h-2)})$. Their oracle actually uses $O(nr^{h-1}+n^2/\sqrt{r})$ space, which leads to the $O(n^{2-1/(2h-1)})$ bound.},
leaving whether such an oracle can be constructed in subquadratic time unsettled.
They also showed how the patterns based on the $\dirLP_{G,M}$ system can be used for computing all the eccentricities of a minor-free digraph in $\Ot\left(n^{2-1/(3h^2+6)}\right)$ time.

We propose a system $\dirMBs_{G,\Delta}$ of \emph{multiball vectors} (discussed in detail in~\Cref{sec:setsystems}) that generalizes both $\dirBs_G$ and $\dirLP_{G,M}$
These objects can be thought of overlaying multiple balls centered at the same vertex with fixed differences in radii.
Our main contribution with this system is predominantly conceptual. 
These multiball vectors are more natural than the Li-Parter set systems used in \cite{LeW24}, are used in the same way for bounding the number of patterns, and are straightforward to provide dimension bounds for.
In particular, we prove that the \emph{Pollard pseudodimension}~\cite{Pollard90} (which constitutes a generalization of the VC dimension) of \emph{multiball vectors} is at most $h-1$. 
This shows the upper bound on the VC dimension of $\dirLP_{G,M}$ is $h-1$ as well, improving the $h^2$ bound shown by~\cite{LeW24}, and removes the discrepancy between the bound they attain in undirected vs directed graphs.
It also immediately implies a tighter running time bound of $\Ot\left(n^{2-1/(3h+3)}\right)$ of the  algorithm for computing eccentricities given in~\cite{LeW24}.

\paragraph*{Faster algorithms for unweighted digraphs.}
While the VC dimension bounds imply limited numbers of distinct patterns per piece, the patterns are still polynomial in size, e.g., of size $\Theta(\sqrt{r})$.
Processing the patterns efficiently is non-trivial and critical if one hopes for subquadratic running time in the end.

We delve into how the distance patterns can be efficiently computed and manipulated.
We observe that representing their piecewise collections using the dynamic strings data structures~\cite{AlstrupBR00, GawrychowskiKKL18, MehlhornSU97} is very helpful to that end.
This idea not only leads to an improved eccentricities algorithm for $K_h$-minor-free digraphs but also enables obtaining first subquadratic bounds for computing the Wiener index and constructing a distance oracle with $O(\polylog(n))$ query. Formally, we show the following:

\begin{restatable}{theorem}{thmunweightedoracle}\label{thm:unweighted-oracle}
Let $G$ be an unweighted digraph that is $K_h$-minor-free. Then, in $\Ot\left(n^{2-1/(3h-2)}\right)$ time one can construct an $O\left(n^{2-1/(3h-2)}\right)$-space exact distance oracle for $G$ with $O(\log{n})$ query time.
\end{restatable}
The oracle behind~\Cref{thm:unweighted-oracle} is the same as that given by~\cite{LeW24} but with the parameter $r$ of the \linebreak $r$-division set differently. 
The oracle of~\cite{LeW24} had $r$ chosen so that the space was optimized, but it was unclear whether such an oracle (in fact, for any choice of $r$) could be constructed in optimal or even subquadratic time.
We show that by sacrificing $o\left(n^{1/(5h)}\right)$ space and processing distance patterns efficiently, we can obtain a subquadratic oracle with near-optimal construction time. The algorithm constructing the oracle of~\Cref{thm:unweighted-oracle} can be extended fairly easily to obtain the following.
\begin{restatable}{theorem}{thmeccentricitieswiener}\label{thm:eccentricities-wiener}
Let $G$ be a $K_h$-minor-free unweighted digraph. Then, in $\Ot\left(n^{2-1/(3h-2)}\right)$ time one can compute (1) all the vertex eccentricities of $G$, and (2) the Wiener index of $G$.
\end{restatable}
Significantly, we show the first subquadratic algorithm for computing the Wiener index in an unweighted minor-free digraph, thus solving the problem left open by~\cite{LeW24}. As~\cite{LeW24} argue, their strategy to compute eccentricities could not be easily applied for the Wiener index problem.

\paragraph*{Decremental reachability oracle and bottleneck paths.}
Surprisingly, we also manage to apply the VC set systems approach to a dynamic data structure problem on directed minor-free graphs. We show:
\begin{restatable}{theorem}{decreachability}\label{thm:dec-reachability}
Let $G$ be a $K_h$-minor-free directed graph.
There exists a deterministic data structure maintaining $G$ subject to edge deletions and supporting arbitrary-pair reachability queries in $O(1)$ time.

The total update time of the data structure is $\Ot(n^{2-1/h})$.
\end{restatable}
Dynamic reachability and its partially dynamic variants (incremental and decremental, allowing either only edge insertions or only edge deletions, respectively) are very well-understood in general graphs from both the upper- and lower-bounds perspective, e.g.~\cite{AbboudW14, BernsteinPW19, BrandNS19, HenzingerKNS15, Italiano86, Roditty08, RodittyZ08, RodittyZ16, Sankowski04}. It has also been studied in planar~\cite{Charalampopoulos22, DiksS07, Karczmarz18, Subramanian93} and minor-free graphs~\cite{Karczmarz18}.
Strikingly, so far, all the non-trivial dynamic data structures for planar and minor-free digraphs that support \emph{arbitrary-pair} reachability queries~\cite{DiksS07, Karczmarz18, Subramanian93} (let alone more general queries, e.g., about distances) have offered only polynomial query times. 
In particular, the techniques used for designing near-optimal static oracles for reachability (dipath separators in~\cite{Thorup04}) and distances (Voronoi diagrams~\cite{CharalampopoulosGLMPWW23}) proved very challenging to be applied in the dynamic setting.
Our data structure is therefore the first to achieve amortized sublinear update time per edge update and polylogarithmic query time at the same time, \emph{even for partially dynamic planar digraphs}, where, recall, much more structure and techniques are available. Observe that for planar digraphs (that are $K_5$-free), the data structure has $\Ot(n^{4/5})$ amortized update time if all edges are eventually deleted.
\Cref{thm:dec-reachability} constitutes the first indication that the VC set systems may also be useful in dynamic settings.

Interestingly, \Cref{thm:dec-reachability} is obtained by applying the bounded VC dimension arguments to the \emph{bottleneck metric} ball system arising from the (unknown till the updates are finished) sequence of graph snapshots. As a by-product, our decremental reachability algorithm can be easily converted into a \emph{bottleneck distance oracle} with $\Ot\left(n^{2-1/h}\right)$ space and preprocessing, and $O(\log{n})$ query time (see~\Cref{l:bottleneck-oracle}). The bottleneck distance oracle is in turn very useful in obtaining an $(1+\eps)$-approximate distance oracle for real-weighted minor-free graphs with the same bounds. We describe that in~\Cref{sec:approx} for completeness, especially since it is unclear whether our weighted exact oracle (\Cref{thm:weighted_do}) can be constructed in subquadratic time at all.

\subsection{Organization}
In~\Cref{sec:setsystems} we introduce our multiball vectors system and relate it to the previously studied set systems. In~\Cref{sec:weighted} we present a randomized variant of our exact distance oracle for real-weighted graphs, which is derandomized only later, in~\Cref{sec:det_weighted_do}. In~\Cref{sec:overview}, we give an overview of our developments for unweighted graphs: subquadratic algorithms for distance problems and the decremental reachability oracle. Then, in~\Cref{sec:pseudodimenson} we prove the pseudodimension bound of the system defined in~\Cref{sec:setsystems}.
\Cref{sec:unweighted} presents our subquadratic algorithms for distance-related problems in unweighted graphs in detail. \Cref{s:decremental} is devoted to the decremental reachability oracle. 

\subsection{Preliminaries}
Throughout this paper, we will be discussing weighted and unweighted $K_h$-minor-free digraphs $G = (V,E)$ with $n=|V|$ vertices. 
For weighted digraphs, we will assume that all weights are positive\footnote{We believe we can extend our results to allow for zero weight edges, but this increases the complexity of the proofs in \Cref{sec:pseudodimenson}.}. 
For two vertices $u,v \in G$, we will let $d_G(u,v)$ denote the shortest path distance between $u$ and $v$. For a piece $P$ of an $r$-division $\rdiv$, for vertices $u,v\in V(P)$ we will sometimes discuss distances in the edge induced subgraph of $P$ and denote such distances as $d_P(u, v)$. For brevity of notation, when it is clear we are discussing distances in the digraph $G$, we will omit subscripts with $d(u,v)$.

\section{Set systems of bounded shattering dimension}\label{sec:setsystems}

Let $G = (V, E)$ be a directed graph. Define a \emph{directed ball centered at $v$ of radius $r$} for a vertex $v\in V$ and $r\in \R$ as $\dirB(v, r) = \{u \in V \mid d(v, u) \le r \}$.
Let us also define $\dirB(u,\infty)$ to be the set of all vertices reachable from
$u$ in $G$. Whenever we talk about a balls in a graph $H\neq G$, we use the notation $\dirB_H(u,r)$. 

Le and Wulff-Nilsen~\cite{LeW24} show this set system has VC dimension at most $h-1$ if $G$ is $K_h$-minor-free:
\[ \dirBs_G = \{\dirB(v, r) \mid v\in V, r \in \R\}. \]
Let $S = \langle s_0, s_1, ..., s_{k-1}\rangle$ be an ordered list of $k$ vertices of $G$. Consider $\ell-1$ many distances \linebreak $\delta_1 < \delta_2 < \cdots < \delta_{\ell-1}$ and let $M =\{\delta_1, \delta_2, ..., \delta_{\ell-1}\} \subseteq \R$.
Le and Wulff-Nilsen also studied sets of the following form for $v\in V$:
\[ \dirX_v = \{ (i, j) \mid d(v, s_i) - d(v, s_0) \le \delta_j  \}. \]
Furthermore, they define $\dirLP_{G, M} = \{ \dirX_v \mid v \in V\}$ where $\dirLP_{G, M}$ is a collection of subsets of $[k-1] \times {[\ell-1]}$.
They showed that when $G$ is $K_h$-minor-free, $\dirLP_{G,M}$ has VC dimension at most $h^2$.

It is difficult to interpret the set $\dirX_v$.
Indeed, even Le and Wulff-Nilsen admit that the definition looks rather complicated~\cite{LeW24}.
We give our own interpretation by relating it to the set system of directed balls.
$\dirX_v$ encodes for each $s_i$ the approximate distance of $v$ to $s_i$ relative to $s_0$ by implicitly storing the unique integer $1\le j \le \ell$ such that
\[ d(v, s_0) + \delta_{j-1} < d(v, s_i) \le d(v, s_0) + \delta_{j}, \]
where we let $\delta_0 = - \infty$ and $\delta_{\ell} = \infty$. For convenience we will define $\Delta = M \cup \{-\infty, \infty\}$.
The above can be interpreted as identifying $s_i$ to be outside the ball $B(v, r_v + \delta_{j-1})$ and within the ball $B(v, r_v + \delta_{j})$,
where~$r_v$ is set to $d(v,s_0)$.
We call $r_v$ the \emph{base radius} since the radii of these balls are relative to $r_v$, with what we call a \emph{shift vector} of different shifts defined by $\Delta$.
Inspired by this interpretation, we now define what we call the \emph{multiball} vector that captures which of the concentric balls points lie in.

\begin{figure}
    \centering
    \includegraphics[page=2, scale=0.6]{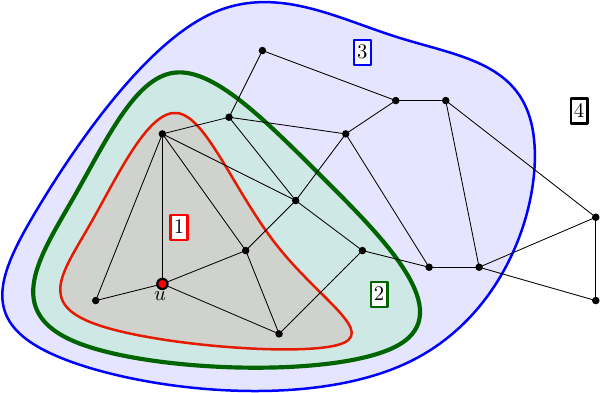}
    \caption{A undirected unweighted graph $G = (V, E)$. Pictured is $\dirMB(u, r, \Delta)$, the multiball centered at $u$ with radius $r=2$ and $\Delta = \{-\infty, -1, 0, 2, \infty \}$.}
    \label{fig:multiball}
\end{figure}

\begin{definition}[Multiballs]\label{def:multiball}
Let $v\in V$ be a vertex of a graph $G$, $r$ a real number, and $\Delta$ be a set of distances $-\infty = \delta_0 < \delta_1 < \cdots < \delta_{\ell-1} < \delta_\ell = \infty$.
Define the multiball vector as follows (see \Cref{fig:multiball}).
\[ \dirMB(v, r, \Delta) =(y_u)_{u\in V} 
    \text{ where } y_u \in [\ell] \text{ is the smallest integer such that } u \in \dirB(v, r+\delta_{y_u}). \]
Let $\dirMBs_{G, \Delta} = \{\dirMB(v, r, \Delta) \mid v\in V, r\in \mathbb{R}\}$ be the set of all multiballs vectors in $G$ with shifts $\Delta$.
\end{definition}

We believe that defining $\dirMBs_{G,\Delta}$ as a set of multivalued functions is more natural than the set system of $\dirLP_{G,M}$. Furthermore, if we view each multiball vector
as a multiball function $f:V\to [\ell]$ and $\dirMBs_{G, \Delta}$ as a collection of these multiball functions, we can leverage richer measurements of the complexity of function classes beyond VC dimension, in particular 
the \emph{pseudodimension} \cite{Pollard90} (sometimes called Pollard dimension) that we provide a definition of here.

\begin{definition}\label{def:pdim}
For a family of functions $\cF$ mapping $V$ to $\{0, \ldots, \ell -1\}$, the \emph{pseudodimension} of $\cF$, denoted by $\Pdim(\cF)$, is the largest integer $d$ such that there exists $x_1, \ldots, x_d \in V$  and thresholds $y_1, \ldots, y_d \in [\ell]$ such that for any $d$-length bit vector $\dirb = (b_1, \ldots, b_d) \in \{0, 1\}^{d}$ there exists an $f\in \cF$ such that:
\[  \forall i, \quad f(x_i) \le y_i \Leftrightarrow b_i = 1 \]
\end{definition}

We remark that $\dirMBs_{G, \Delta}$ is simultaneously a generalization of $\dirBs_G$ and of $\dirLP_{G, \Delta}$, and proving bounds on the pseudodimension of $\dirMBs_{G,\Delta}$ immediately gives bounds on the VC dimension of $\dirBs_G$ and $\dirLP_{G,\Delta}$. This is formalized in the following two observations.
\begin{observation}
$\VCdim(\dirBs_G) = \Pdim(\dirMBs_{G,\{-\infty, 0, \infty\}})$ 
\end{observation}
\begin{proof}
Observe the multiball $\dirMB(v, r,\{-\infty, 0, \infty\})\in \dirMBs_{G,\{-\infty, 0, \infty \}}$ gives a vector in $\{1, 2\}^V$. This vector exactly encodes the ball $\dirB(v, r) \in \dirBs_G$, so there is a bijection between $\dirMBs_{G, \{-\infty, 0, \infty\}}$ and $B_G$. The definition of pseudodimension for $\ell = 2$ exactly matches the definition of VC dimension.
\end{proof}

Consider $\dirX_v \in \dirLP_{G,M}$.
There is a correspondence between $\dirX_v$ and the vector $\dirMB_{\mid S}(v, d(v, s_0), \Delta)$, the restricted multiball vector to the set of vertices $S\subseteq V$ with a base vertex $s_0 \in S$. 
This shows that the family of multiballs $\dirMBs_{G, \Delta}$ are much more expressive than $\dirLP_{G,M}$, they contain information about all vertices of $G$. The balls in the family also allow for arbitrary base radii, not just $d(v, s_0)$.
We can formalize the expressiveness of multiball vectors with the following observation.
\begin{observation}
$\VCdim(\dirLP_{G, M}) \le \Pdim(\dirMBs_{G, \Delta})$ with $\Delta = M\cup \{-\infty, \infty\}$.
\end{observation}
\begin{proof}
Let $\{(i_1, y_1),\ldots,(i_h, y_h)\}$ be a shattered set of $\dirLP_{G,M}$ with size $k$.
Consider \Cref{def:pdim} with $\{x_1, \ldots, x_k\} = \{s_{i_1}, \ldots, s_{i_k}\}$ and $\{y_1, \dots, y_k\}$. 
For any $k$-length bit vector $\vec{b} = \{b_1, \ldots, b_k\}$, we can find a set $X_{v} \in \dirLP_{G,M}$ where $b_j = 1$ iff $(i_j, y_j) \in X_v$.
This shows that $\Pdim(\dirMBs_{G, \Delta})$ is at least $k$.
\end{proof}

Our main result in this section is the following theorem where we establish a bound of $h-1$ on the pseudodimension of multiball vectors. The proof we will defer until \Cref{sec:pseudodimenson}, as it is almost identical to the minor-building proofs used to bound VC dimension used by \cite{LeW24}.
Our main contribution here is predominantly conceptual; we believe these multiball vectors are more natural set systems to consider.

\begin{restatable}{theorem}{mbpseudodim}

\label{thm:mb_pseudodim} 
For a $K_h$-minor-free graph $G$, for any set $\Delta$ of $\ell$ distances, the set of multiball vectors $\dirMBs_{G, \Delta}$ has pseudodimension at most $h-1$.
\end{restatable}

Thus we obtain an improved version of the results obtained by Le--Wulff-Nilsen~\cite{LeW24}.
\begin{corollary}
\label{cor:vcdim} 
For a $K_h$-minor-free graph $G$, for any set $\Delta$ of $\ell$ distances, the set systems $\dirB_G$ and $\dirLP_{G, \Delta}$ have VC dimension at most $h-1$.
\end{corollary}

We can use the following generalization of the Sauer-Shelah Lemma 
to multivalued functions proven by Haussler and Long~\cite{HausslerL95}, which gives polynomial growth of multivalued functions in terms of both the size of the ground set and the number of different values the functions can take.

\begin{theorem}[Generalized Sauer-Shelah Lemma. Corrolary~3~in~\cite{HausslerL95}] \label{thm:gen_sauer}
Let $\cF$ be a family of functions mapping a ground set with $n$ elements to $[\ell]$.
If the pseudodimension of $\cF$ is at most $d$, then $|\cF| \le O(n^d\ell^d)$.
\end{theorem}

The following lemma is a generalization
of a well-known fact concerning the VC dimension.

\begin{lemma}\label{lem:pdim-restrict}
 Let $\cF$ be a family of functions as in~\Cref{thm:gen_sauer}. Let $S$
 be any subset of the ground set of~$\cF$. Call $\cF_S=\cF\cap {[\ell]}^S$ a \emph{restriction} of $\cF$ to $S$. Then, $\Pdim(\cF_S)\leq \Pdim(\cF)$. 
\end{lemma}
\begin{proof}
The lemma follows simply from the fact that, in the restricted case, the elements $x_1,\ldots,x_d$ in~\Cref{def:pdim} maximizing $d$ are allowed to be chosen only from
a subset $X$ of the ground set.
\end{proof}
We will often use the following combination of \Cref{thm:gen_sauer}~and~\Cref{lem:pdim-restrict}.
\begin{corollary}\label{cor:restrict}
Let $\cF$ be a family of functions as in~\Cref{thm:gen_sauer} and let $S$ be any subset of the ground set of $\cF$. Let $d=\Pdim(\cF)$. Then
$|\cF_S|=O(|S|^d\ell^d)$.
\end{corollary}

\section{Exact distance oracles in weighted digraphs}\label{sec:weighted}

In this section, we describe the first application of VC dimension techniques to directed weighted graphs with a randomized version of \Cref{thm:weighted_do}. Specifically, we show that given a \emph{weighted} $K_h$-minor-free digraph $G=(V,E)$ on $n$ vertices, where $h=O(1)$, there exists a subquadratic-space \emph{randomized} exact distance oracle that can handle queries in $O(\log n)$ expected time. We derandomize this distance oracle in \Cref{sec:det_weighted_do}.

First, we compute an $r$-division $\rdiv$ of $G$, and for every two vertices $u,v\in V(P)$ of each piece $P\in \rdiv$, we precompute the distance $d_G(u,v)$ between them.

Suppose we wish to answer a distance oracle query between a vertex $s\in V$ and a vertex $t\in V$.
Suppose that $t$ lies in a piece $P\in \mathcal{R}$.
Since distances between vertices inside $P$ are precomputed, it remains to handle the case when $s\notin V(P)$.  
We know that the shortest path between $s$ and $t$ must pass through one of the $\ell = O(\sqrt{r})$ vertices $v_1,\dots, v_\ell \in \bnd P$.
Thus, the length of the shortest path satisfies the following equation.
\[d(s, t) = \min_{1\le i \le \ell} \left(d(s, v_i) + \dist_P(v_i, t)\right)\]

A crucial idea of our algorithm to find the index $i^*$ that attains the minimum is the so-called Fredman's trick \cite{Fredman76} that has been used extensively for faster algorithms for the APSP problem \cite{Fredman76,Chan10,Williams18,Chan20,ChanWX23}.
The so called ``trick'' is that for two indices $i$ and $j$, determining whether \linebreak $d(s, v_i) + d(v_i, t) <  d(s, v_j) + d(v_j, t)$ is equivalent to testing:
\begin{equation}\label{eq:fredman}
    d(s, v_i) -  d(s, v_j) <  d(v_j, t) - d(v_i, t). 
\end{equation} 
While this is nothing more than a simple rearrangement of the inequality, observe that the right hand side now
involves distances solely within the piece $P$ between a vertex inside $P$ and two vertices on the boundary~$\bnd{P}$ of which there are only $O(r^2)$ possibilities.
Furthermore, the left-hand side is a difference of distances from $s$ to two boundary vertices.

We combine Fredman's trick with an algorithm reminiscent of quick select to quickly zoom in on the vertex on the boundary of a piece that is on the shortest path between $s$ and $t$. Specifically, we will reduce the problem of finding the shortest path to a problem about finding the minimum element in a hidden permutation $\sigma = (\sigma_1, \sigma_2, \dots, \sigma_\ell)$ that sorts the paths through the boundaries: 
\[ d(s, v_{\sigma_1}) + d(v_{\sigma_1}, t) 
\le d(s, v_{\sigma_2}) + d(v_{\sigma_2}, t)
\le \quad \cdots  \quad
\le d(s, v_{\sigma_\ell}) + d(v_{\sigma_\ell}, t).
\]
Our data structure constructs an oracle that given an index $x$ where $x = \sigma_i$, we can obtain a pointer to an unordered set $\{\sigma_1, \sigma_2, \dots, \sigma_{i-1}\}$ stored in an array that we have random access to. Thus, to find $\sigma_1$, it suffices to solve the following problem.

\begin{problem}[Min-finding problem]\label{prob:min_finding}
Let $\sigma_1, \dots, \sigma_n$ be a hidden permutation of the numbers from $1$ to~$n$. Suppose we have an oracle that takes as a query a number $x$, and returns a pointer to an unsorted set $S_x$ containing the elements $\{\sigma_1, \dots, \sigma_{i-1}\}$ where $\sigma_i = x$ (if $\sigma_1 = x$ then $S_x$ is the empty set).
The goal is to find $\sigma_1$ in the minimum number of queries.
\end{problem}

There is a straightforward algorithm that repeatedly makes random queries to the set that contains $\sigma_1$. With probability $1/2$, a random choice will decrease the size of the set by at least $1/2$. By a standard Chernoff bound, we will find $\sigma_1$ in at most $O(\log n)$ queries with high probability.

\paragraph*{Constructing the distance oracle.}
Given an $r$-division $\rdiv$,
for every vertex $s\in V$, piece $P\in\rdiv$, and every vertex $v\in\bnd{P}$, we compute and store the distance from $s$ to $v$ in $G$. Also, as indicated earlier, for every pair of vertices $u, v \in V(P)$ we also store the distances from $u$ to $v$ in a lookup table.

Fixing a vertex $v_j$ on the boundary of a piece $P$, we define the following shift vector that is motivated by Fredman's trick: 
\[\Delta_j = \{d(v_j, t) - d(v_i, t) \mid v_i \in \bnd{P}, t \in V(P) \}\cup \{-\infty,\infty\}.\]
Now for any vertex $s\in V$, consider the vector $Y_{s,j}$ that is the restriction of $\dirMB(s,d(s, v_j), \Delta_j)$ to $\bnd{P}$, where
\[ Y_{s,j}[i] = \dirMB(s,d(s, v_j), \Delta_j)[v_i].\]
For every $s\in V$ and every $j\in [\ell]$ we can store a pointer to $Y_{s,j}$.
Now for any fixed vertex $t\in V(P)$, the vector $Y_{s,j}$ encodes enough information to determine which of the other vertices of $v_i\in \bnd P$ have a shorter path from $s$ to $t$ through $v_i$ than the path through $v_j$. 
To be precise, knowing $Y_{s,j}$ allows us to compute the following set of vertices:
\[  \{v_i \in \bnd{P} \mid d(s, v_i) - d(s, v_j) < d(v_j, t) - d(v_i, t)\} := X_t[Y_{s, j}]. \]
At vertex $t$, we can store all this information in a dictionary $X_t$ that takes as input a pointer to $Y_{s, j}$ and outputs this set of vertices through which there are shorter paths.

\paragraph*{Randomized distance query.} 
Given a query from vertex $s \in V$ to a vertex $t$ that lies in piece $P$ that has boundary vertices $v_1,\dots, v_\ell \in \bnd{P}$, the following is a randomized algorithm for finding the distance between $s\not\in V(P)$ and $t$ (recall that if $s\in V(P)$, the solution is stored in a lookup table).

\begin{enumerate}
    \item Pick any vertex $v_i \in \bnd P$ as a guess of the vertex the shortest path from $s$ to $t$ passes through.
    \item To check if $v_i$ is indeed a vertex on the shortest path from $s$ to $t$, we check whether $X_t[Y_{s, i}]$ is empty by passing the pointer to $Y_{s, i}$ that we have stored at $s$ to the dictionary $X_t$ stored at $t$.
    \item  If $X_t[Y_{s, i}]$ is empty, we conclude that $v_i$ is on the shortest path from $s$ to $t$ and we can output $d(s, v_i) + d(v_i, t)$ as the answer.
    Otherwise, we pick a random vertex from within $v_{i'}\in X_t[Y_{s, i}]$ and repeat the previous step with $v_{i'}$ instead of $v_i$.
\end{enumerate}

\paragraph*{Correctness.}
Observe that the algorithm described is exactly the algorithm for the min-finding problem (\Cref{prob:min_finding}). The hidden permutation we are looking for the minimum of is the sorted list of \linebreak $d(s,v_i) + d(v_i, t)$, and the dictionary $X_t$ implements the oracle that outputs the set.

\paragraph*{Query runtime.}
The randomized algorithm to the min-finding problem takes $O(\log \ell)$ queries in expectation and with high probability.
Since we only follow $O(1)$ pointers for each query and $\ell = O(\sqrt{r})$, the query algorithm takes $O(\log r)$ expected time.

\paragraph*{Space analysis.} 
To store all distances from vertices to boundaries of the $r$-division takes total space $O(n^2/\sqrt{r})$, and the size of the lookup tables between distances of pairs of vertices in the same piece is $O(nr)$.

Let us now fix a piece $P\in\rdiv$.
Since $|\Delta_j|=O(r^{3/2})$, the number of different vectors $Y_{s,j}$ for a fixed $j\in [\ell]$ is at most $O\left(\left(\sqrt{r}\cdot r^{3/2}\right)^{h-1}\right) = O(r^{2(h-1)})$ by~\Cref{thm:mb_pseudodim}~and~\Cref{cor:restrict}, and thus at most $O(r^{2h-3/2})$ per piece as $\ell = O(\sqrt{r})$.
To store the map $X_t$, we need to store a pointer to a subset of vertices of the boundary $\bnd{P}$ which takes $O(\sqrt{r})$ space.
Since there are $O(r)$ maps $X_t$,
their total size is $O(r^{2h})$.

By summing through all pieces $P\in\rdiv$< we obtain that the total space used
is $O\left(n^2/\sqrt{r}+nr^{2h-1}\right)$.
Setting $r = n^{2/(4h-1)}$ yields a space usage of $O(n^{2-1/(4h-1)})$ and a query time of $O(\log n)$.

\section{Technical overview: unweighted applications}\label{sec:overview}
Let $G$ be an unweighted $K_h$-minor-free directed graph. Recall the template of~\Cref{sec:template}; all our algorithmic developments for unweighted graphs follow that template.
Let~$\rdiv$ be an $r$-division of $G$. In the following, we generally focus on computing distances (or reachability) between two vertices located in distinct pieces of~$\rdiv$. The number of single-piece pairs is only $O(nr)$, so they are easier to handle.

\subsection{Faster algorithms for unweighted digraphs}

\paragraph{Base patterns.} For constructing unweighted distance oracles, computing the eccentricities and the Wiener index, the main set system that we use (and was also used by~\cite{LeW24} in their distance oracle, reiterated in~\Cref{sec:unweighted-oracle}) is $\dirBs_G$.
The patterns for a piece $P\in \rdiv$ are the intersections of balls from $\dirBs_G$ with~$V(P)$.
Recall that since $\dirBs_G$ has VC dimension $h-1$ (\Cref{cor:vcdim}), the number of patterns per piece is $O(r^{h-1})$ (\Cref{cor:restrict}).

For each pair $(s,P)$ such that $s\notin V(P)$, we would like to store not one (as in the running example from~\Cref{sec:template}), but $|\bnd{P}|$ pattern pointers.
Namely, we would like to store pointers to each of the patterns
$\dirB(s,\dist(s,b))\cap V(P)$ for $b\in \bnd{P}$. As observed by~\cite{LeW24}, these pointers are useful for the following reason.

Suppose $b_1,\ldots,b_k$ are the boundary vertices of $P$ sorted non-decreasingly
by $\dist(s,b_i)$.
Let $t\in V(P)$.
Let~$j$ be the largest index such that $t\notin \dirB(s,\dist(s,b_j))$.
If $j$ does not exist, then $t\in \dirB(s, \dist(s, b_1))$. However, $b_1$ must be the closest vertex in $P$ from $s$. In this case we can easily determine that $d(s, t) = d(s, b_1)$.
Otherwise, one can prove that
\begin{equation}\label{eq:ovw1}
\dist(s,t)=\dist(s,b_j)+\dist_P\left(\dirB(s,\dist(s,b_j))\cap V(P),t\right).
\end{equation}
Consequently, by~\eqref{eq:ovw1}, for any $t\in V(P)$, the distance between $s$ and $t$ can be computed based on:
\begin{enumerate}[label=(\arabic*)]
    \item The distances from $s$ to $\bnd{P}$ in $G$ (a part of the ``global'' data),
    \item The distance in $P$ from a pattern $\dirB(s,\dist(s,b_j))\cap V(P)$ and $t\in V(P)$ (which is a part of the ``local'' data for the piece $P$).
\end{enumerate}
Note that the piecewise local data can be computed in $O((n/r)\cdot r^{h-1}\cdot r)=\Ot(nr^{h-1})$ time, whereas the global data takes $\Ot(n^2/\sqrt{r})$ time to compute by running Dijkstra's algorithm to each $\bnd{\rdiv}$.

As shown by~\cite{LeW24}, given the global and the local data, and the pointers to the stored patterns, one can use~\eqref{eq:ovw1} to answer a distance query in $O(\log{n})$ time since the index $j$ can be found via binary search using the local data for piece $P$. The formula~\eqref{eq:ovw1} allows computing the eccentricities easily as well, and is also the base of our algorithm for computing the Wiener index (\Cref{sec:wiener}).

The challenge in constructing the unweighted distance oracle in subquadratic time (and also computing the Wiener index) lies in computing the pointers to the $O(\sqrt{r})$ patterns 
$\dirB(s,\dist(s,b))\cap V(P)$, $b\in \bnd{P}$, and generating the required patterns themselves.
This is also the element of the template of~\Cref{sec:template} that the oracle of~\cite{LeW24} was missing and thus a non-trivial construction algorithm was not given.

\paragraph{Proxy patterns.} The high-level idea to deal with the challenge is to use another layer of ``proxy'' patterns.
Fix a piece $P$. Suppose we want to find the required base pattern pointers for all $s\in V\setminus V(P)$.

We use the following observation. For any $x\in\mathbb{N}$, the pattern of the form $\dirB(s,\dist(s,b'))\cap V(P)$ is uniquely determined by
$P$, the value $\dist(s,b')-x$, and the vector $(\dist(s,b)-x)_{b\in \bnd{P}}$. Intuitively, this is because every path from $s$ to $V(P)$ has to pass through some last vertex in $\bnd{P}$, and thus the ball $\dirB(s,d(s,b'))\cap V(P)$ can be expressed as the union of balls $\dirB_P(b,\dist_G(s,b')-\dist_G(s,b))$
inside $P$, for $b\in \bnd{P}$. At the same time, each radius $\dist(s,b')-\dist(s,b)$ can be computed based on the values $\dist(s,b')-x$ and $\dist(s,b)-x$ only.

Consider first a special case when $P$ is strongly connected. Then, for any two boundary vertices $b,b'\in \bnd{P}$, we have $|\dist(s,b)-\dist(s,b')|<r$. By our observation above, if we set $b^*\in\bnd{P}$ to be the boundary vertex of $P$ with the minimum distance from $s$, then the required balls $\dirB(s,\dist(s,b))\cap V(P)$ are uniquely determined by the vector
\begin{equation}\label{eq:ovw2}
\delta_{P,s}:=\left(\dist(s,b)-\dist(s,b^*)\right)_{b\in \bnd{P}}\in [r]^{\bnd{P}},
\end{equation}
which constitutes the other (``proxy'') type of pattern that we use. To see that there cannot be too many patterns of this kind, note that each pattern~\eqref{eq:ovw2} can also be seen as an element of the system $\dirLP_{G,[r]}$ or a multiball vector (see~\Cref{sec:setsystems})
$\dirMB\left(s,\dist_G(u,b^*),[r]\cup\{\pm\infty\}\right)$ from $\dirMBs_{G,[r]\cup\{\pm\infty\}}$ restricted to $\bnd{P}$. That set, in turn, has size $O((|\bnd{P}|\cdot r)^{h-1})=O(r^{3(h-1)/2})$ by~\Cref{thm:mb_pseudodim}~and~\Cref{cor:restrict}.

Observe that for each pair $(s,P)$, the pattern $\delta_{P,s}$ can be obtained from the ``global'' data in $O(\sqrt{r})$ time.
For each of the $O(r^{3(h-1)/2})$ encountered patterns for $P$, the $O(\sqrt{r})$ originally required balls can be computed in $O(\poly(r))$ time. This is why the vectors $\delta_{P,s}$ can be called proxy patterns.

\paragraph{The general case and dynamic strings.} The strong connectivity assumption is very powerful and allows mapping each $(s,P)$ to a \emph{single} proxy pattern.
To handle the general case, we first observe (\Cref{l:interval-ball-unique}) that if the distances $\dist_G(s,b)$, for $b\in \bnd{P}$, are spread apart too much, then we do not need all the exact offsets $\dist(s,b')-\dist(s,b)$ to uniquely determine $\dirB(s,b')\cap V(P)$.
In fact, we only need the exact offsets for those $b\in\bnd{P}$ whose distances from $s$ are $O(r)$ far from $\dist(s,b')$.

Still, if the distances in question are too far apart, it is unclear how to leverage the above observation to handle the general case with a single proxy pattern per $(s,P)$ pair. In fact, it is not even clear whether $o(\sqrt{r})$ such patterns can be enough. This is problematic since the proxy patterns are inherently $\Theta(\sqrt{r})$-space objects. If we were to compute all of them explicitly, the running time would be quadratic in $n$.

Our technical idea to deal with this problem is to define the proxy patterns in a way that guarantees that the $O(\sqrt{r})$ patterns for $(s,P)$ do not differ too much so that these patterns can be represented implicitly. 
Indeed, the proxy patterns from $\{-r,\ldots,r\}^{\bnd{P}}$ that we end up using for $(s,P)$ are such that each subsequent of them can be obtained from the previous one using amortized $O(1)$ coordinate changes.
Consequently, if we view these patterns as $|\bnd{P}|$-length strings over the alphabet of size $O(r)$, then we can compute a fingerprint representation (allowing $O(1)$-time equality tests) of them all in $\Ot(\sqrt{r})$ time using dynamic strings data structures~\cite{AlstrupBR00, GawrychowskiKKL18, MehlhornSU97}.
This allows computing all the $O\left(n^2/\sqrt{r}\right)$ proxy pattern fingerprints in $\Ot\left(n^2/\sqrt{r}\right)$ total time (through all pairs $(s,P)$).

\paragraph{Speeding up ball computation.} We also find a further application of the dynamic strings data structure.
Namely, generating all the original patterns $\dirB(s,\dist(s,b))\cap V(P)$ from the encountered proxy patterns requires testing whether a ball of this kind was not yet generated and stored. Since each such a ball is uniquely determined by a proxy pattern, we could generate the $O(\sqrt{r})$ balls for a proxy pattern naively in $O(r^{3/2})$ time. 
However, these balls are nested and therefore we can compute their fingerprints in $\Ot(r)$ time by taking advantage of a dynamic strings data structure once again.

\subsection{Decremental reachability}
\paragraph{The data structure.} The ``global'' data in our decremental reachability data structure consists of $O(n/\sqrt{r})$ decremental single-source reachability (SSR) data structures (run on the reverse graph $G^R$ for each possible source in $\bnd{\rdiv}$) maintaining reachability between all pairs $V\times \bnd{\rdiv}$ subject to deletions in $G$. In the basic randomized variant of our data structure, one can use a near-optimal decremental SSR data structure~\cite{BernsteinPW19} with near-linear total update time. Then, maintaining the global data subject to deletions issued to $G$ takes $\Ot(n^2/\sqrt{r})$ time in total.

Fix a piece $P\in\rdiv$. The patterns we use for $P$ are the subsets of $\bnd{P}$ that are currently reachable in $G$ from some $s\in V$.
Note that for a fixed $s\in V$, the subset of $\bnd{P}$ reachable from $s$ -- which we hereby define to be the pattern $\delta_{P,s}$ -- can only shrink with time.
As a result, each pattern $\delta_{P,s}$ changes only $O(\sqrt{r})$ times.

The local data for a pattern $\delta$ for $P$ is the current set of $V(P)$ reachable from the subset $\delta\subseteq \bnd{P}$ via a path in~$P$. This local data can be also maintained using a decremental SSR data structure on $P$ (with an auxiliary super-source modelling $\delta$) so that the total time spent per pattern $\delta$ is $\Ot(r)$.

Note that the global and local data together can be used to answer any reachability query $(s,t)$ where $s\in V\setminus V(P)$ and $t\in V(P)$ if only we can maintain a pointer from the pattern $\delta_{P,s}$ to the local data of $P$ associated with that pattern. 
To guarantee that, we again leverage the dynamic strings data structure to map the patterns into $O(1)$-size fingerprints that we can use as an addressing layer for the patterns of $P$.

Since $\delta_{P,s}$ undergoes $O(\sqrt{r})$ updates in total, its fingerprint can be maintained at all times in $\Ot(\sqrt{r})$ total time.
Whenever a fingerprint unseen before is encountered, the appropriate local data for it (i.e., the decremental SSR data structure inside a piece) is initialized.

\paragraph{Analysis and bottleneck paths.} The total update time of the data structure can be seen to be $\Ot\left(n^2/\sqrt{r}\right)$ plus $\Ot((n/r)\cdot r)=\Ot(n)$ times the total number of patterns encountered for a single piece $P$.
We analyze this quantity by considering the entire (unknown) sequence of updates and observing that in fact every pattern encountered is a restriction of some \emph{bottleneck ball} to $\bnd{P}$ in a related graph $G'$ whose edge weight constitute timestamps of the respective edge deletions.

For a weighted digraph $H$, the \emph{bottleneck length} of a path $u\to v$ is defined to be the maximum edge weight on that path. The \emph{bottleneck distance} $\beta_{H}(u,v)$ between $u$ and $v$ is defined as the minimum bottleneck length of a $u\to v$ path in $H$. The bottleneck ball of $s$ in $H$ is the set of vertices of $H$ with bottleneck distance from $s$ within some threshold. Bottleneck distances are less general than standard distances: for any $H$ they can be encoded using the latter by mapping the $i$-th smallest weight of $H$ to, e.g., $n^i$.
As a result, the bottleneck balls of $K_h$-minor-free graphs also have VC dimension at most $h-1$ by~\Cref{cor:vcdim}.

Since bottleneck balls in $G'$ have VC dimension $h-1$, we obtain that for each piece, at most $O\left(|\bnd{P}|^{h-1}\right)=O\left(r^{(h-1)/2}\right)$ different patterns can be observed for any order of edge deletions issued to the initial graph $G$. We conclude that the total update time of the data structure is $\Ot\left(n^2/\sqrt{r}+nr^{(h-1)/2}\right)$. Setting $r=n^{2/h}$ yields a total update time of $\Ot(n^{2-1/h})$.

\paragraph{Derandomization.} The best-known deterministic decremental SSR data structures for general graphs~\cite{BernsteinGS20} are not as efficient as the randomized alternative~\cite{BernsteinPW19}. However, many of the used decremental SSR instances in our application are run on related graphs.
As a result, we can use a certain decremental low-diameter reachability emulator data structure for minor-free graphs from~\cite{Karczmarz18} so that we can replace decremental SSR data structures for general graphs with very simple near-linear deterministic decremental SSR data structures for $\polylog(n)$-diameter digraphs~\cite{EvenS81, HenzingerK95}. This incurs no additional overhead cost.

\paragraph{Further applications.} Our data structure can be very easily turned into a bottleneck distance oracle for weighted graphs with space and construction time matching the total update time $\Ot(n^{2-1/h})$ of the decremental reachability oracle, and $O(\log{n})$ query time. This in turn can be used to obtain a strongly polynomial $(1+\eps)$-approximate (standard) distance oracle for real-weighted digraphs with the same preprocessing, space, and query bounds (see~\Cref{sec:approx} for details). 

\section{Pseudodimension of multiballs in minor-free graphs}\label{sec:pseudodimenson}
Let $G = (V, E)$ be a directed graph where the edges have real-weights.
For simplicity, we will assume that all edges have positive weights;
zero weight edges can be handled as well.\footnote{For example, for the purpose of analysis one could reweight each edge $e$ to $(n/\delta)\cdot w(e) + 1$, where $\delta$ is the smallest possible positive difference between weights of two simple paths in $G$. Let $G'$ be the obtained positively weighted digraph. Note that every ball $\dirB(v,r)$ (or, more generally, multiball vector) in $G$ has a corresponding ball $\dirB(v,(n/\delta)\cdot r+(n-1))$ (multiball vector) in $G'$. Thus, proving a bound on the pseudodimension of multiball vectors in $G'$ implies the same bound in $G$. 
}
We show that if the pseudodimension of a graph is $d$, then the graph must contain $K_d$ as a minor.
The proof uses the minor-building argument employed by Le and Wulff-Nilsen~\cite{LeW24} (they attribute the proof structure to Chepoi, Estellon, and Vaxes~\cite{ChepoiEV07}).
They (implicitly) prove the following lemma which we call the minor building lemma, that requires certain conditions on shortest paths. They actually require an additional property (Claim~1 in \cite{LeW24}) that we show is unnecessary in our proof of the lemma. We defer the proof to \Cref{ap:minor_building}.

\begin{restatable}[Minor building via shortest paths~\cite{LeW24}]{lemma}{minorbuilding}
\label{lem:minor_building}
Let $G = (V, E)$ be a graph. Let $\pi(x,y)$ denote any shortest path between $x$ and $y$.
Suppose we have a set of vertices $v_1, ..., v_d$ such that for every $1\le i < j \le d$
there is a vertex $t_{ij}$ that satisfies the following:
\begin{enumerate}
    \item[(1)]  For all distinct $i,j,p,q\in [d]$, $\pi(t_{ij} , v_i)$ and $\pi(t_{pq}, v_p)$ are vertex disjoint.
    \item[(2)]  For all distinct $i,j,p\in [d]$,
    $V(\pi(t_{ij} , v_i)) \cap V(\pi(t_{jp} , v_j)) \subseteq \{t_{ij}\}$.
\end{enumerate}
Then $G$ contains $K_d$ as a minor. 
\end{restatable}

Let $\Delta$ to be a vector of $\ell-1$ real weights $\delta_1 < \delta_2 < \cdots < \delta_{\ell-1}$ augmented by $\{-\infty,\infty\}$.
Suppose that the pseudodimension of $\dirMBs_{G,\Delta}$ is~$d$.
This means there exists a set of vertices $v_1, ..., v_d \in V$, and a set of threshholds $y_1, ..., y_d \in [\ell]$ 
such that for every pair $\{i, j\}\subseteq [d]$, there exists a $t_{ij}\in V$ and an $r_{ij} \in \R$ such that the function $f_{ij} = \dirMB(t_{ij}, r_{ij}, \Delta) \in \dirMBs_{G, \Delta}$ satisfies:
\begin{itemize}
    \item[(a)] $f_{ij}(v_k) \le y_k$ for $k \in \{i, j\}$, 
    meaning that $d(t_{ij}, v_k) \le r_{ij} + \delta_{y_k}$. 
    \item[(b)] $f_{ij}(v_k) > y_k$ for all $k \in [\ell] - \{i, j\}$,
    meaning that $d(t_{ij}, v_k) > r_{ij} + \delta_{y_k}$. 
    \item[(c)] $r_{ij}$ is minimum among all choices of pairs $(t_{ij}, r_{ij})$.
\end{itemize}

We show that these $t_{ij}$ satisfy the conditions of \Cref{lem:minor_building}, and so $G$ must contain $K_d$ as a minor. We claim that the conditions of \Cref{lem:minor_building} are satisfied in the following.

\begin{figure}[hb!]
    \centering
    \begin{subfigure}{0.3\textwidth}
    \includegraphics[page=2]{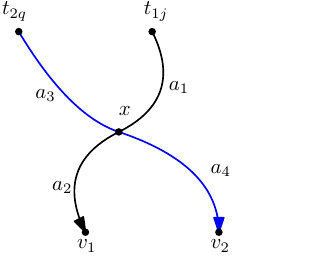}
    \caption{Figure for \Cref{claim2}.}
    \label{fig:minor_building2}
    \end{subfigure}
    \begin{subfigure}{0.3\textwidth}
    \includegraphics[page=3]{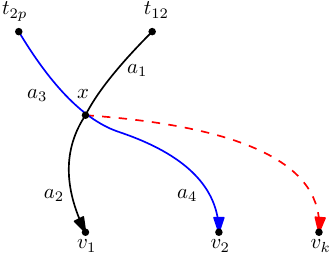}
    \caption{Figure for \Cref{claim3}.}
    \label{fig:minor_building3}
    \end{subfigure}
    \caption{Illustrations for claims.}
\end{figure}

\begin{claim} \label{claim2}
For all distinct $i,j,p,q\in [d]$, any shortest paths $\pi(t_{ij} , v_i)$ and $\pi(t_{pq} , v_p)$ are vertex disjoint.
\end{claim}

\begin{proof}
Without loss of generality, we can renumber the indices so that $i=1$ and $p=2$.
Suppose to the contrary that $x\in V(\pi(t_{1j} , v_1)) \cap V(\pi(t_{2q} , v_2))$.
See \Cref{fig:minor_building2} for an illustration.
Define the following quantities:
\begin{align*}
    a_1 &= d(t_{1j} , x) \quad &  a_3 &= d(t_{2q} , x) \\
    a_2 &= d(x , v_1) \quad &  a_4 &= d(x , v_2) 
\end{align*}

Note that $a_1+a_2 = d(t_{1j} , v_1)$ and $a_3+a_4 = d(t_{2q} , v_2)$ as $x$ is on the shortest path, and by the triangle inequality we have that $a_1 + a_4 \ge d(t_{1j} , v_2)$ and $a_2 + a_3 \ge d(t_{2q} , v_1)$. Combined with the definitions of $t_{1j}$ and $t_{2q}$ and their associated radii, the following mix of equalities and inequalities hold.
\begin{align*}
a_1 + a_2 = d(t_{1j}, v_1) &\le r_{1j} + \delta_{y_1}   
\quad   &a_1 + a_4 \ge d(t_{1j}, v_2) > r_{1j} + \delta_{y_2}    \\
a_3 + a_4 = d(t_{2q}, v_2) &\le r_{2q} + \delta_{y_2}   
\quad   &a_2 + a_3 \ge d(t_{2q}, v_1) > r_{2q} + \delta_{y_1}    
\end{align*}

Putting together the inequalities from above we derive the following.
\[
(r_{1j} + \delta_{y_1}) + (r_{2q} +\delta_{y_2}) 
\ge a_1 + a_2 + a_3 + a_4  
> (r_{1j} + \delta_{y_2}) + (r_{2q} + \delta_{y_1}).
\]
This is a contradiction, thus $x$ cannot exist.
\end{proof}

\begin{claim} \label{claim3}
For all distinct $i,j,p\in [d]$, $V(\pi(t_{ij} , v_i)) \cap V(\pi(t_{jp} , v_j)) \subseteq \{t_{ij}\}$.
\end{claim}

\begin{proof}
Without loss of generality, we can renumber $i,j,p$ so that $i=1$ and $j=2$, and $p\in [d] \setminus \{1, 2\}$.
Suppose to the contrary that $x\in V(\pi(t_{12} , v_1)) \cap V(\pi(t_{2p} , v_2))$ with $x\neq t_{12}$.
See \Cref{fig:minor_building3} for an illustration.
Define the following quantities:
\begin{align*}
    a_1 &= d(t_{12} , x) \quad &  a_3 &= d(t_{2p} , x) \\
    a_2 &= d(x , v_1) \quad &  a_4 &= d(x , v_2) 
\end{align*}

We show 
that if we define $r_x = r_{12}-a_1$, then $g  = \dirMB(x, r_x, \Delta) \in \dirMBs_{G, \Delta}$ satisfies Properties (a) and~(b). Since the edge weights are positive, $a_1>0$, and thus $r_x < r_{12}$. But this contradicts Property (c).

We begin by making some similar observations as in \Cref{claim2}, that $a_1+a_2 = d(t_{12} , v_1)$ and $a_3+a_4 = d(t_{2p} , v_2)$ as $x$ is on the shortest path, and by the triangle inequality we have that $a_2 + a_3 \ge d(t_{2p} , v_1)$. Thus the following mix of equalities and equalities hold.
\begin{align*}
a_1 + a_2 = d(t_{12}, v_1) &\le r_{12} + \delta_{y_1}   
\quad 
\\
a_3 + a_4 = d(t_{2p}, v_2) &\le r_{2p} + \delta_{y_2}   
\quad   &a_2 + a_3 \ge d(t_{2p}, v_1) > r_{2p} + \delta_{y_1}    
\end{align*}
The first inequality implies that:
\[ a_2 = (a_1 + a_2) - a_1 \le r_{12} + \delta_{y_1} - a_1 = r_x + \delta_{y_1} \]
Combining the above with the previous two inequalities yields:
\[ a_4 = a_2 + (a_3 + a_4) - (a_2 + a_3) < (r_x + \delta_{y_1}) + (r_{2p} + \delta_{y_2}) - (r_{2p} + \delta_{y_1}) =  r_x + \delta_{y_2} \]

Thus, Property (a) holds for $g$.
Now consider Property (b). Since Property (b) holds for $f_{12}$, for all $k\in [\ell] - \{1, 2\}$, $d(t_{12}, v_k) > r_{ij} + \delta_{y_k}$. By the triangle inequality,
\[ 
d(x , v_k) \ge d(t_{12}, v_k) - d(t_{12} , x) > r_{12} + \delta_{y_k} - a_1 = r_x + \delta_{y_k}
\]
This concludes that Property (b) also holds for $g$, which is a contradiction, so $x$ cannot exist.
\end{proof}

Claim~\ref{claim2}, and Claim~\ref{claim3} together satisfy the hypothesis of \Cref{lem:minor_building}. If the graph $G$ were $K_h$-minor free, we can conclude that the pseudodimension of $G$ is strictly less than $h$, so \Cref{thm:mb_pseudodim} holds.

\mbpseudodim*

\section{Deterministic weighted distance oracle} \label{sec:det_weighted_do}

To make the randomized distance oracle from \Cref{sec:weighted} deterministic, we consider a slightly different solution to the min-finding problem (\Cref{prob:min_finding}) where we pick all our randomness upfront by fixing a single random permutation $\tau$. We describe it as follows:

\begin{center} \textbf{Min-finding algorithm with a permutation $\tau$} \end{center}
\begin{enumerate}
    \item Pick a permutation $\tau_1, \tau_2, ..., \tau_n$ of the numbers from $1$ to $n$.
    \item Let $S$ be a set that is guaranteed to contain $\sigma_1$.
    Let $i$ be the minimum value such that $\tau_i\in S$.
    Query the oracle to obtain the set $S_{\tau_i}$.
    \item Repeat the above step until $S$ is empty.
\end{enumerate}

The next lemma shows that the algorithm terminates in $O(\log n)$ rounds in expectation with good tail bounds by using standard backwards analysis \cite{HarPeled11}. 
\begin{lemma} \label{lem:perm_bounds}
If $\tau$ is a random permutation, the above algorithm will succeed in the min-finding problem with at most $c\log n$ queries with probability at least $1-1/n^{\Omega(c)}$ for any $c>1$.
\end{lemma}
\begin{proof}
Observe that the algorithm is equivalent to randomly shuffling a permutation of length $n$, considering the permutation in order and counting the number of times the minimum number seen so far decreases. 

Let $X$ be the total number of times the minimum changes.
Let $X_i$ denote the event that the minimum among the first $i$ elements of the permutation occurs at exactly the $i$th element. Note that these events are independent.
Observe that $X = \sum_{i=1}^n X_i$.
Since $\E[X_i] = 1/i$, the expected number times the minimum changes by linearity of expectation is $\E[X] = \sum_{i=1}^n \E[X_i] = O(\log n)$.
To show the desired tail bound, applying Chernoff bound for $c\ge 1$ yields:
\[ P(X \ge c\, \E[X]) \le \exp\left( - (c-1)^2\cdot O(\log n)/(c+1) \right) = 1/n^{\Omega(c)}.\] 
Thus the probability of the algorithm terminating within $c\log n$ queries is at least $1-1/n^{\Omega(c)}$.
\end{proof}

For our deterministic distance oracle, we will do the same precomputation as in the randomized distance oracle, and in addition for every vertex $t$ in a piece $P\in \rdiv$, we will precompute a fixed permutation $\tau(t)$ of size $|\bnd P|$.
To handle a query between $s\in V$ and a vertex $t\in P$, we will use $\tau(t)$ as the permutation for the min-finding algorithm.
By \Cref{lem:perm_bounds}, if we choose a random permutation for $\tau(t)$ of size $|\bnd P|$, the probability that this permutation works in at most $c\log r$ queries for all vertices $s\in V$ is at least $1-n/r^{\Omega(c)}$ by the union bound.
Since we chose $r = n^{O(1)}$, for sufficiently large $c$, a random permutation guarantees $c\log n$ query time for all vertices $s\in V$ with positive probability. Hence we conclude that some fixed permutation guarantees this as well. 

\begin{theorem}
There exists a deterministic distance oracle for a weighted $K_h$-minor-free digraph $G$ on $n$ vertices that uses $O(n^{2-1/(4h-1)})$ space and can handle queries in $O(\log n)$ time.
\end{theorem}

\paragraph*{Remarks.} 
The proof of the existence of the permutations $\tau(t)$ is non-constructive. 
This can be made into a constructive algorithm with the method of conditional expectations, though construction of the permutations will take superquadratic (although still polynomial) time.
Note that even if we could construct all these permutations in subquadratic time, we do not know how to reduce the precomputation time needed in the rest of the data structure to subquadratic.

\section{Exact distances in unweighted digraphs}\label{sec:unweighted}
In this section, we describe the applications of \Cref{thm:mb_pseudodim} to distance problems in an \emph{unweighted} $K_h$\nobreakdash-minor-free digraph $G=(V,E)$, where $h=O(1)$: subquadratic algorithms computing (1) the exact distance oracle of~\cite{LeW24}, (2) the $n$ vertex eccentricities (thus also the diameter), and (3) the Wiener index.

We start by describing some infrastructure common to all these applications. First of all,
we compute an $r$-division $\rdiv$ of $G$. 
Recall  that this takes $O(n^{1+\eps})$ time. 

Next, we compute and store the following:
\begin{itemize}
    \item the distances $\dist_G(u,v)$ for every $u\in V$ and $v\in \bnd{R}$. They can be obtained by running breadth-first-search on the reverse $G$ from every $v\in \bnd{\rdiv}$. This takes $O(|E|\cdot |\bnd{\rdiv}|)=O(n^2/\sqrt{r})$ time and space.
    \item for every piece $P\in\rdiv$, and every pair of vertices $u,v\in V(P)$, the distances $\dist_P(u,v)$ and $\dist_G(u,v)$. All the ``intra-piece'' distances of the former kind can be computed in $O(|\rdiv|\cdot r^2)=O(nr)$ time.   
    Note that
\[ \dist_G(u,v) = \min\left(\dist_P(u,v), \min_{b\in \bnd{P}} (\dist_G(u,b)+\dist_P(b,v)) \right),\] so the desired distances of the latter kind can be found in $O(nr^{3/2})$ time based on what we have computed previously.
The space used is $O(nr)$.
\end{itemize}

\subsection{The exact distance oracle of~\cite{LeW24}}\label{sec:unweighted-oracle}
\newcommand{\pcbset}{\mathcal{Z}}
Note that the precomputed distances allow answering distance queries $(u,v)$ in $G$ except when $v$ is a non-boundary vertex of a piece $P$ not containing $u$.
\cite{LeW24} deal with those remaining most challenging queries by storing, for each such $u,P$,
the (pointers to) ball-piece intersections
\[\dirB\left(u,\dist(u,b)\right)\cap V(P)\] for all $b\in \bnd{P}$.
By~\Cref{cor:vcdim}~and~\Cref{cor:restrict}, the number of distinct
sets of the above form for a fixed piece $P$ is $O(r^{h-1})$. 
Call the set of such distinct balls $\pcbset_P$.
The balls of $\pcbset_P$ can be
stored explicitly using $O(r^h)$ space (as each ball contains $O(r)$ vertices of $P$),
for a total of $O(nr^{h-1})$ space
over the $n/r$ pieces $P\in\rdiv$. Storing the balls explicitly
enables constant time membership queries in these balls.

Moreover, for each ball-piece intersection from $\pcbset_P$, and $v\in V(P)\setminus \bnd{P}$, the distance
\[\dist_P\left(\dirB(u,\dist_G(u,b))\cap V(P),v\right)\]
is computed and stored. The time and space used for that is also bounded by $O(nr^{h-1})$. 

As~\cite{LeW24} show, the stored ball-piece intersections along with the precomputed distances are enough to answer distance queries in $O(\log{r})$ time. Let us reiterate their argument for completeness.

Suppose $b_1,\ldots,b_k$ are the boundary vertices of $P$ sorted non-decreasingly
by $\dist(u,b_i)$.
Using binary search, in $O(\log{r})$ time one can locate the largest $j$ such that $v\notin \dirB(u,\dist(u,b_j))$ -- recall that membership queries take $O(1)$ time.
The index $j$ along with the preprocessed data
can be used to handle a distance query using the following:
\begin{lemma}
Let $P$ be a piece and $u,v\in V$ be vertices such that $v\in P\setminus \bnd P$ and $u\notin P$. Let $j$ be the largest index such that $v\notin \dirB(u,\dist(u,b_j))$. Then
\begin{equation}\label{eq:distance-query}
\dist_G(u,v)=\dist_G(u,b_j)+\dist_P\left(\dirB(u,\dist_G(u,b_j))\cap V(P),v\right).
\end{equation}
\end{lemma}
\begin{proof}

By the definition of $j$, either $j=k$ or $j<k$ and
$\dist_G(u, b_{j+1}) \ge \dist_G(u,v)>\dist_G(u,b_j)$.
Either way, the shortest $u\to v$ path in $G$ (if it exists) does not pass through $\{b_{j+1},\ldots,b_k\}$: when $j<k$, it follows by $v\notin \bnd{P}$.

Now, if no $u\to v$ path in $G$ exists, then the right-hand side of~\eqref{eq:distance-query} is clearly infinite and thus~\eqref{eq:distance-query} holds.

Let us suppose $\dist_G(u,v)<\infty$ and let $\pi$ be a shortest $u\to v$ path in $G$. $\pi$ can be expressed as $\pi_1\pi_2$, where $\pi_2=c\to v$ and $c$ is the last boundary vertex of $P$ appearing on $\pi$. Note that $\pi_2$ is fully contained in $P$.
Moreover, $\dist_G(u,c)\leq \dist_G(u,b_j)$ as otherwise $c\in \{b_{j+1},\ldots,b_k\}$ which, as we already argued, cannot happen.
Consequently, the vertex $z$ on $\pi$ at distance $\dist_G(u,b_j)$ from $u$ lies on $\pi_2$ which in turn implies $z\in V(P)$. The quantity on the right-hand side of~\eqref{eq:distance-query} is thus no larger than $\dist_G(u,z)+\dist_P(z,v)=\dist_G(u,v)$.

To finish the proof, it is enough to argue that the RHS of~\eqref{eq:distance-query} cannot be smaller than $\dist_G(u,v)$.
Let $y\in \dirB(u,\dist_G(u,b_j))\cap V(P)$ be such that $\dist_P(y,v)$ is minimized. 
Recall that we have
$\dist_G(u,y)\leq \dist_G(u,b_j)$.
Moreover, $\dist_P(y,v)\geq 1$ since $\dist_G(u,v)>\dist_G(u,b_j)\geq \dist_G(u,y)$. 
So if $\dist_G(u,y)=\dist_G(u,b_j)$, then $\dist_G(u,b_j)+\dist_P(y,v)\geq \dist_G(u,v)$ follows by the triangle inequality. We now argue that $\dist_G(u,y)<\dist_G(u,b_j)$ cannot happen. Indeed,  
if this was the case, the vertex $y'$ following $y$ on a shortest $y\to v$ path in $P$ would satisfy $\dist_G(u,y')\leq \dist_G(u,b_j)$ and
$\dist_P(y',v)<\dist_P(y,v)$. This would contradict our choice of $y$.
\end{proof}

The distance oracle uses $O\left(n^2/\sqrt{r}+nr^{h-1}\right)$ space. 
This bound is optimized for $r=n^{\frac{2}{2h-1}}$ and thus the obtained space bound is $O\left(n^{2-\frac{1}{2h-1}}\right)$.
However, it is not clear how the small VC dimension of directed balls in $G$ alone can be leveraged to compute the desired ball-piece intersections $\pcbset_P$ efficiently. 

\subsection{Computing the balls in subquadratic time}
In this section, we describe an algorithm computing the distance oracle of~\cite{LeW24}.
The algorithm runs in subquadratic time for some choice of $r$, proving that
a non-trivial exact distance oracle can be constructed in subquadratic time.
Note that the distance oracle we construct will use more space than the one in~\cite{LeW24}.

Fix a piece $P\in\rdiv$. First, note that any $\dirB(u,q)\cap V(P)$ is uniquely
determined by $q$, $P$, and the distance
vector $(\dist_G(u,b))_{b\in \bnd{P}}$.
This is formally captured by the following lemma:
\begin{lemma}\label{lem:piece-ball-union}
Let $u\in V\setminus V(P)$ and $q\in\mathbb{R}$. Then $\dirB(u,q)\cap V(P)=\bigcup_{b\in\bnd{P}} \dirB_P(b,q-\dist_G(u,b))$.
\end{lemma}
\begin{proof}
Let $v\in \dirB(u,q)\cap V(P)$.
Since $u\notin V(P)$, the shortest $u\to v$ path has to pass through
some last boundary vertex $b^*$ of $P$. That is, $\dist_G(u,v)=\dist_G(u,b^*)+\dist_P(b^*,v)$.
This implies $\dist_P(b^*,v)\leq q-\dist_G(u,b^*)$ and thus $v\in \dirB_P(b,q-\dist_G(u,b^*))$.
The $\text{RHS}\subseteq \text{LHS}$ inclusion is trivial.
\end{proof}
In fact, from~\Cref{lem:piece-ball-union} a more general property follows:
that $\dirB(u,q)\cap V(P)$ is uniquely determined by $P$, $q-x$, and the vector
$(\dist_G(u,b)-x)_{b\in \bnd{P}}$, where $x$ is a possibly unknown real offset value.
\begin{lemma}\label{lem:piece-dijsktra}
    Let $u\in V\setminus V(P)$. Then, the distances $\dist_G(u,v)$ for all $v\in V(P)$ can be computed in $\Ot(r)$ time based on $P$ and the distances $\dist_G(u,b)$ given for all $b\in \bnd{P}$. 
\end{lemma}
\begin{proof}
For $v\in V(P)$, any shortest $u\to v$ path passes through some last from $\bnd{P}$,
and thus $\dist_G(u,v)=\min_{b\in \bnd{P}}(\dist_G(u,b)+\dist_P(b,v))$.
Observe that such minima can be found
by running Dijkstra's algorithm on the graph $P$ augmented with a super-source $s$
added and edges $sb$ of weight $\dist_G(u,b)$ for all $b\in \bnd{P}$.
\end{proof}

\subsubsection{A special case}
For $u\in V\setminus V(P)$, let $b_{u,0}$ be the boundary vertex of $P$ that
minimizes $q_u:=\dist_G(u,b_{u,0})$ (if one exists).
To illustrate the overall approach, consider first a special case when
for all $u\in V$ and $b\in \bnd{P}$ reachable from $u$,
\begin{equation}\label{eq:dist-oracle-easy}
    \dist_G(u,b)-q_u< r.
\end{equation}
This is the case, for example, 
if the piece $P$ is strongly connected.
Observe that, under the above assumption, the distances $\dist_G(u,b)$ are uniquely determined by $q_u$ and the position $y_b$
of the first ball containing $b$ in
the sequence
\[\dirB(u,q_u+(\dist_G(u,b_{u,0})-q_u)+1)=\dirB(u,q_u+1),\dirB(u,q_u+2),\ldots,\dirB(u,q_u+(r-1)).\]
Note that $\{y_b\}_{b\in\bnd{P}}$ is precisely the multiball vector
$\dirMB(u,q_u,\Delta)$, where
$\Delta=\{1,\ldots,r-1,\infty\}$, that belongs to the restriction 
of $\dirMBs_{G,\Delta}$ to $\bnd{P}$.
Therefore, by~\Cref{thm:mb_pseudodim}~and~\Cref{cor:restrict}, the number of such distinct multiball vectors (through all $u$ and $q_u$)
is $O((|\bnd{P}|\cdot r)^{h-1})=O(r^{3(h-1)/2})$.
We conclude that for a fixed $P\in\rdiv$, the number of distinct relative boundary distance vectors from $u$ to $\bnd{P}$:
\[ \delta_{P,u}:=\left(\dist_G(u,b)-\dist_G(u,b_{u,0})\right)_{b\in \bnd{P}}, \]
called \emph{patterns} in the following,
through all $u$, is $O(r^{3(h-1)/2})$ as well.

\newcommand{\bdvdict}{\mathcal{D}}

We can use the patterns (and their bounded number) as a proxy for constructing the 
pointers from $u\in V$ to the required balls, as follows.
For each piece $P$, we store an (initially empty) dictionary $\bdvdict_P$ mapping the already encountered patterns $(z_b)_{b\in \bnd{P}}$ to the collection of corresponding ball-piece intersections from $\pcbset_P$ required by the distance oracle of~\cite{LeW24}.
That is, for each $c\in \bnd{P}$, we store a pointer $\bdvdict_P(z)[c]$ to
\begin{equation}\label{eq:bpieceunion}
\bigcup_{b\in \bnd{P}}\dirB_P(b,z_c-z_b)
\end{equation}
Note that if $\delta_{P,u}$ is a pattern from $u$ to $\bnd{P}$, and $c\in \bnd{P}$, then by~\Cref{lem:piece-ball-union}:
\begin{align*}
    \dirB(u,\dist_G(u,c))\cap V(P) &= \bigcup_{b\in \bnd{P}}\dirB_P(b,\dist_G(u,c)-\dist_G(u,b))\\
    &=\bigcup_{b\in \bnd{P}}\dirB_P\left(b,\left({(\delta_{P,u})}_c+\dist_G(u,b_{u,0})\right)-\left({(\delta_{P,u})}_b+\dist_G(u,b_{u,0})\right)\right)\\
    &=\bigcup_{b\in \bnd{P}}\dirB_P\left(b,{(\delta_{P,u})}_c-{(\delta_{P,u})}_b\right)\\
    &=\bdvdict_P(\delta_{P,u})[c].
\end{align*}

Having fixed a piece $P$, we process the sources $u\in V\setminus V(P)$ one by one.
We can construct the pattern $\delta_{P,u}$ in $\Ot(\sqrt{r})$ time based on the stored distances.
If $\delta_{P,u}\notin \bdvdict_P$, then we compute and store the $O(\sqrt{r})$
sets~$Z$ of the form~\eqref{eq:bpieceunion}, one for each $c\in \bnd{P}$.
Note that each such set $Z$ consists of some number of vertices of $P$ that are
closest to $u$; the distances from $u$ to $V(P)$ in $G$ can be found in $\Ot(r)$ time
by~\Cref{lem:piece-dijsktra}.
As a result, each $Z$ is computed in $\Ot(r)$ time.
Next,~$Z$ is inserted into $\pcbset_P$ if it was
not present there yet -- this can be tested in $\Ot(|Z|)=\Ot(r)$ time. Finally, $\bdvdict_P(\delta_{P,u})[c]$ is set to point to
the location of $Z$ in $\pcbset_P$.

Once the pattern $\delta_{P,u}$ is present in $\bdvdict_P$,
we simply extract the $|\bnd{P}|$ pointers to $\pcbset_P$
from $\bdvdict_P(\delta_{P,u})$, as required by the distance oracle, in $O(\sqrt{r})$ time.

Note that the described computation when $\delta_{P,u}\notin \bdvdict_P$ takes $\Ot(r^{3/2})$ time ($\Ot(r)$ time for each set of the form~\eqref{eq:bpieceunion}), but can only happen $O(r^{3(h-1)/2})$ times. As a result, the total time cost
incurred through all $O(n)$ sources $u$ is $O(n\sqrt{r}+r^{3h/2})$.
Over all pieces $P\in\rdiv$, the total time spent is $\Ot\left(n^2/\sqrt{r}+nr^{3h/2-1}\right)$.

\subsubsection{The general case}
We will now explain how to extend the above approach to deal with the general case when  condition~\eqref{eq:dist-oracle-easy} does not hold.
The main challenge here is
that it is unclear how to establish a correspondence between the previously considered patterns and multiball vectors with a shifts set of size $\poly(r)$.
We circumvent this problem by using multiple patterns per piece-source pair $(P,u)$. In fact, the number of patterns used will possibly be $\Theta(|\bnd{P}|)$. This causes further problems, if the used patterns have size $|\bnd{P}|$, as before: if one computed all the patterns naively and explicitly, their total description length could be $\Theta(n^2)$.

Let $P\in\rdiv$ and $u\in V\setminus V(P)$.
Once again, suppose $\bnd{P}=\{b_1,\ldots,b_k\}$ so that
$\dist_G(u,b_1)\leq \ldots\leq \dist_G(u,b_k)$.
We start with the following lemma.

\begin{lemma}\label{l:interval-ball-unique}
    Let $s\in \{1,\ldots,k\}$. Let $i$ be the smallest index such that $\dist_G(u,b_i)> \dist_G(u,b_s)-r$.
    Let $j$ be the largest index such that $\dist_G(u,b_j)< \dist_G(u,b_s)+r$.
    Then, for all $l=s,\ldots,j$, the set $\dirB\left(u,\dist_G(u,b_l)\right)\cap V(P)$ is
    uniquely determined by $P$ and the values $i$, $j$, and $(\dist_G(u,b_i)-\dist_G(u,b_s)),\ldots,(\dist_G(u,b_j)-\dist_G(u,b_s))$.
\end{lemma}
\begin{proof}
Let $v\in V(P)$. Note that $l\in  [s,j]$ implies that $\dist_G(u,b_l)<\infty$.
We prove that we can determine whether $v$ is in 
$\dirB(u,\dist_G(u,b_l))$ based only on the piece~$P$ and the values $l$, $i$, $j$, and $(\dist_G(u,b_i)-\dist_G(u,b_s)),\ldots,(\dist_G(u,b_j)-\dist_G(u,b_s))$. 
First, if $v$ is not reachable from $\bnd{P}$ in $P$, then clearly $v\notin \dirB(u,\dist_G(u,b_l))$
(recall that $u\notin V(P)$).

If $v$ is reachable from some $b_t\in \{b_1,\ldots,b_{i-1}\}$
in $P$, then
\[\dist_G(u,v)\leq \dist_G(u,b_t)+\dist_P(b_t,v)< \dist_G(u,b_t)+r\leq \dist_G(u,b_s)\leq \dist_G(u,b_l).\]
So $v\in \dirB(u,\dist_G(u,b_l))$. 
Let us then assume that $v$ is not reachable from $\{b_1,\ldots,b_{i-1}\}$.
Then, any shortest $u\to v$ path in $G$ has some last vertex in
$b_t\in \{b_i,\ldots,b_k\}$. Hence, $v\in \dirB(u,\dist_G(u,b_l))$ iff
${\dist_G(u,b_t)+\dist_P(b_t,v)\leq \dist_G(u,b_l)}$ for some
$b_t\in \{b_i,\ldots,b_k\}$. Note that, equivalently, the inequality
has to hold for some ${b_t\in \{b_i,\ldots,b_j\}}$, since for $t'\in \{j+1,\ldots,k\}$,
we have
\[\dist_G(u,b_{t'})+\dist_P(b_{t'},v)\geq \dist_G(u,b_{t'})>\dist_G(u,b_j)\geq \dist_G(u,b_l).\]
Whether the required $t$ exists can be tested by checking, for all $t\in \{i,\ldots,j\}$, the inequality\linebreak
$(\dist_G(u,b_t)-\dist_G(u,b_s))+\dist_P(b_t,v)\leq (\dist_G(u,b_l)-\dist_G(u,b_s))$.
\end{proof}

Equipped with Lemma~\ref{l:interval-ball-unique}, we now define the set of patterns
used for the pair $(P,u)$. Let the \emph{pivot indices} $\Pi_{P,u}\in [k]$ for $\bnd{P}$ be defined
inductively as follows:
\begin{itemize}
    \item the index $1$ is included in $\Pi_{P,u}$,
    \item for subsequent $i=2,\ldots,k$, the index $i$ is included in $\Pi_{P,u}$ if the last previously
    included $j\in\Pi_{P,u}$ (where $j<i)$ satisfies $\dist_G(u,b_j)\leq\dist_G(u,b_i)-r$.
\end{itemize}
Let $\Pi_{P,u}=\{\pi_1,\ldots,\pi_{\ell}\}$, where $\pi_1\leq\ldots\leq\pi_{\ell}$. For convenience, introduce sentinels $\pi_0:=0$ and $\pi_{\ell+1}:=k+1$.
The size of $\Pi_{P,u}$ may generally range from $1$ to $|\bnd{P}|$.

We define $\ell$ patterns $\delta_{P,u,1},\ldots,\delta_{P,u,\ell}$ coming from $\{-r,\ldots,r\}^{\bnd{P}}$. For $b_t\in \bnd{P}$, we have:
\[ \left(\delta_{P,u,i}\right)_{b_t} = \begin{cases}
    \max(-r, \dist_G(u,b_t)-\dist_G(u,b_{\pi_i})) & \text{ if }t<\pi_{i+1},\\
    \min(r, \dist_G(u,b_t)-\dist_G(u,b_{\pi_i})) & \text{ if }t\geq \pi_{i+1}.
\end{cases}\]
Note that in the first case above for $t<\pi_{i+1}$ we indeed have:
\[\dist_G(u,b_t)-\dist_G(u,b_{\pi_i})\leq \dist_G(u,b_{\pi_{i+1}-1})-\dist_G(u,b_{\pi_i})\leq r-1< r.\]
From Lemma~\ref{l:interval-ball-unique} we can infer the following.
\begin{lemma}\label{lem:ball-pattern}
    The balls $\dirB(u,\dist_G(u,b_l))\cap V(P)$ for indices $l$ satisfying $\pi_q\leq l<\pi_{q+1}$ are all uniquely determined by $P$ and the pattern $\delta_{P,u,q}$.
\end{lemma}
\begin{proof}
We apply Lemma~\ref{l:interval-ball-unique} with $s:=\pi_q$. Then, the index $i$ in Lemma~\ref{l:interval-ball-unique} equals $\pi_{q-1}+1$ by the definition of $\pi_q$. Similarly,  the index $j$ in Lemma~\ref{l:interval-ball-unique} equals $\pi_{q+1}-1$ by the definition of $\pi_{q+1}$. For all $y=s,\ldots,j$ we have $\dist_G(u,b_y)-\dist_G(u,b_s)<r$ and analogously for all
$y=i,\ldots,s$ we have $\dist_G(u,b_y)-\dist_G(u,b_s)>-r$.
So for all $y=i,\ldots,j$, $(\delta_{P,u,q})_{b_y}=\dist_G(u,b_y)-\dist_G(u,b_s)$.
Moreover, we have $(\delta_{P,u,q})_{\pi_{q-1}}=-r$ and  $(\delta_{P,u,q})_{\pi_{q+1}}=r$.
This means that all the data required by~\Cref{l:interval-ball-unique} (besides $P$) to uniquely determine the desired balls can be read from the pattern $\delta_{P,u,q}$.
\end{proof}
As in the special case, a pattern $\delta_{P,u,i}$ corresponds 
to the unique multiball vector $\dirMB(u,\dist_G(u,b_{\pi_i}),\Delta)$, where
now we use $\Delta=\{-\infty,-r,\ldots,r,\infty\}$, that in turn belongs to the restriction 
of $\dirMBs_{G,\Delta}$ to $\bnd{P}$.
Once again, by~\Cref{thm:mb_pseudodim}~and~\Cref{cor:restrict}, we conclude that the numbers
of distinct patterns possible is $O((|\bnd{P}|\cdot |\Delta|)^{h-1})=O(r^{3(h-1)/2})$.

We would like to proceed analogously as in the special case. For each piece $P$, store an (initially empty) dictionary $\bdvdict_P$ mapping the already encountered patterns $(z_b)_{b\in \bnd{P}}$ to the collection of corresponding ball-piece intersections from~$\pcbset_P$ so that many pointers
of the distance oracle can be inferred from the previous computation.
Since now we might have as many as $|\bnd{P}|$ patterns per $(P,u)$ pair, we cannot
afford to spend $\widetilde{\Theta}(|\bnd{P}|)$ time for dictionary operations and even for constructing each of the patterns $\delta_{P,u,i}$. We deal with both these problems
by employing a variant of the dynamic strings data structure, defined as follows.

\newcommand{\sid}{\operatorname{id}}
\newcommand{\tcoll}{\mathcal{T}}
\newcommand{\tcollb}{\mathcal{T}^{\mathcal{B}}}

\begin{definition}
    Let $\Sigma$ be an alphabet. A \emph{dynamic strings data structure} maintains an initially empty collection $\tcoll$ of distinct strings from $\Sigma^*$. It supports adding new strings to the collection, and each insertion of a string $s\in \Sigma^*$ returns an \emph{identifier} $\sid(s)$ of $s$. Two strings from $\tcoll$ have equal
    identifiers if and only if they are equal. The two allowed insertion operations are:
    \begin{enumerate}
        \item insertion of an explicitly given string $s\in\Sigma^*$,
        \item given the identifier $\sid(s)$ of some $s\in\tcoll$, an index $k\in [|s|]$, and $c\in \Sigma$, insertion of the string $s'$ obtained from $s$ by replacing its $k$-th letter with symbol $c$.
    \end{enumerate}
\end{definition}
\begin{theorem}\label{t:dyn-string}{\upshape \cite{AlstrupBR00, MehlhornSU97}}
    There exists a deterministic dynamic strings data structure that processes a sequence of string insertions in $O(t\polylog{t})$ time,
    where $t$ equals $|\Sigma|$ plus the total number of insert operations plus the total length of the strings inserted explicitly (using the insertion of the first kind). The produced identifiers are $O(\polylog(t))$-bit integers.\footnote{It is worth noting that widely known dynamic strings data structures~\cite{AlstrupBR00, GawrychowskiKKL18, MehlhornSU97} support a much more general interface that allows producing new strings by concatenations and splits (at a desired position) of previously inserted strings that we do not need. Our dynamic string operations can be easily reduced to those. There also exist more efficient~\cite{AlstrupBR00, GawrychowskiKKL18} or simpler (Karp-Rabin fingerprint-based) randomized dynamic strings data structures.}
\end{theorem}
We set up a fresh dynamic strings data structure $\tcoll_P$ (with alphabet $\Sigma=\Delta\setminus\{-\infty,\infty\}$) of Theorem~\ref{t:dyn-string} for each piece $P$. Its main goals are (1) speeding up the construction of the multiple patterns for each $u$, and (2) compression of the pattern keys in the dictionary $\bdvdict_P$. 
Our algorithm will guarantee
that the quantity~$t$ in Theorem~\ref{t:dyn-string} is polynomial in $r$, so that
the produced identifiers can be manipulated in $\Ot(1)$ time.

For a fixed $u\in V\setminus V(P)$, the pivots $\Pi_{P,u}$ can be obtained in $\Ot(\sqrt{r})$ time after sorting the boundary vertices $|\bnd{P}|$ by their distance from $u$.
Having found the pivots, we start by explicitly setting up the zeroth ``sentinel'' pattern $\delta_{P,u,0}=(r)^{|\bnd{P}|}$ and inserting it into $\tcoll_P$; this way we obtain its identifier $\sid(\delta_{P,u,0})$.

Observe that each subsequent (identifier of) $\delta_{P,u,i}$ can be produced from $\delta_{P,u,i-1}$ by performing some number of coordinate changes,
each of which decreases some coordinate (character) $c$ either (1) all the way to $-r$,
or (2) by at least $r$ if $c\in [0,r)$, or (3) to a value less than $r$ if $c=r$.
Consequently, when transforming $\delta_{P,u,0}\to \delta_{P,u,1}\to\ldots\to \delta_{P,u,\ell}$,
every individual of the $k$ coordinates undergoes at most three changes:
its starting value is $r$, it may drop to some value between between $-r$ and $r$ exactly once, and after it becomes $-r$ the coordinate never changes again.
Therefore, if the strings
after each coordinate change are inserted into~$\tcoll$, the identifiers of all $\delta_{P,u,i}$
can be obtained using $O(|\bnd{P}|)$ operations on the data structure $\tcoll_P$, i.e., in $O(|\bnd{P}|\polylog{r})=\Ot(\sqrt{r})$ time. These identifiers are then
used for querying the dictionary $\bdvdict_P$ whether the relevant balls determined
by these patterns have already been computed. Recall that if a certain ball has not been yet encountered, we can compute it explicitly in $\Ot(r)$ time based on the pattern (see Lemma~\ref{l:interval-ball-unique}), and store a pointer to it in $\bdvdict_P$. Recall also that there are $O(\sqrt{r})$ relevant balls per pattern.

As in the special case, the total time spent through all pieces is $\Ot\left(n^2/\sqrt{r}+nr^{3h/2-1}\right)$.

\subsubsection{Speeding up balls computation with dynamic strings}\label{s:speeding-balls}
The bottleneck of the computation so far was computing the $O(\sqrt{r})$ balls of the form
$\dirB(u,\dist_G(u,b_l))\cap V(P)$, for $l=s,\ldots,j$, based on a single pattern $\delta$ storing the values $\dist_G(u,b_i)-\dist_G(u,b_s),\ldots,\dist_G(u,b_j)-\dist_G(u,b_s)$
that, together with $P$, uniquely determine these balls  (see~\Cref{l:interval-ball-unique} and~\Cref{lem:ball-pattern}).
Each of these balls was computed naively in $O(r)$ time, and its presence in $\pcbset_P$ was tested. This incurred
an $\Ot(r^{3/2})$ time overhead per encountered pattern.
Note, however, that the balls $\dirB(u,\dist_G(u,b_l))\cap V(P)$ are not unrelated,
as we have
\begin{equation}\label{eq:balls-inclusion}
    \dirB(u,\dist_G(u,b_s))\cap V(P)\subseteq \dirB(u,\dist_G(u,b_{s+1}))\cap V(P)\subseteq \ldots\subseteq \dirB(u,\dist_G(u,b_j))\cap V(P).
\end{equation}
Recall from the proof of Lemma~\ref{l:interval-ball-unique} that $v\in \dirB(u,\dist_G(u,b_l))$
if and only if $v$ is reachable from $\{b_1,\ldots,b_{i-1}\}$ in $P$ or 
there exists $t\in \{i,\ldots,j\}$ such that $(\dist_G(u,b_t)-\dist_G(u,b_s))+\dist_P(b_t,v)\leq (\dist_G(u,b_l)-\dist_G(u,b_s))$.
But the quantities $\min_{t=i}^j\{(\dist_G(u,b_t)-\dist_G(u,b_s))+\dist_P(b_t,v)\}$
can be easily computed for all $v\in V(P)$ in $\Ot(r)$ time using a single run of Dijkstra's algorithm on $P$ with a supersource
(as in Lemma~\ref{lem:piece-dijsktra}). Therefore, in $\Ot(r)$ time one can easily determine,
for each $v\in V(P)$,
the first ball from the sequence~\eqref{eq:balls-inclusion} that contains $v$.

This observation can be used as follows. Let us additionally store all the balls encountered (for the piece~$P$) in a dynamic strings data structure $\tcollb_P$ of~\Cref{t:dyn-string} with $\Sigma:=\{0,1\}$.
That is, assuming some total order on $V(P)$, the strings in $\tcollb_P$ correspond
to subsets of $V(P)$. Let again $\pcbset_P$ store the required ball-piece intersections explicitly, but now keyed
by their identifiers from $\tcollb_P$.
Observe that, when processing a pattern $\delta$, the identifiers of the balls from the sequence~\eqref{eq:balls-inclusion} can be computed using
$O(r)$ dynamic strings operations, as follows. We start from the string $\phi$ corresponding to the empty subset of $V(P)$. For each $i=s,\ldots,j$ we toggle in $\phi$ the coordinates corresponding to all $v\in V(P)$ such that $\dirB(u,\dist_G(u,b_i))$ is the earliest ball in~\eqref{eq:balls-inclusion} containing $v$. This way, each coordinate in $\phi$ is toggled at most once.
For each of the considered balls, if its identifier
is not yet in $\pcbset_P$, we insert its explicit representation with the appropriate
key to $\pcbset_P$.

The discussed optimization reduces the time needed to compute all the required
balls for ${P\in\rdiv}$ to $\Ot\left(\left(|\pcbset_P|+|\bdvdict_P|\right)\cdot r\right)=\Ot(r^{(3h-1)/2})$. The total time spent processing all pieces is thus $\Ot\left(n^2/\sqrt{r}+nr^{3(h-1)/2}\right)$.
This bound is optimized for $r=n^{\frac{2}{3h-2}}$, and thus we obtain:
\thmunweightedoracle*
It is worth recalling that the space-bound of the distance oracle remains $O\left(n^2/\sqrt{r}+nr^{h-1}\right)$. The most space-efficient exact distance oracle that we can still construct in
truly subquadratic time is obtained when~$r$ is maximum such that $nr^{3(h-1)/2}$ is still subquadratic. Specifically, for $r=n^{\frac{2}{3h-3}-2\eps}$ the preprocessing is truly subquadratic and the data structure
uses $\Ot\left(n^{2-\frac{1}{3h-3}+\eps}\right)$ space.

\subsection{Eccentricities and the Wiener index}\label{sec:wiener}
The data stored by the distance oracle is powerful enough to enable (with minor adjustments) computing
the vertex eccentricities and the Wiener index within the distance oracle's construction time.

\paragraph{Eccentricities.} Recall that the eccentricity of a vertex $u\in V$ is defined as $e(u)=\max_{v\in V}\{\dist_G(u,v)\}$. Since $O(n^2/\sqrt{r}+nr)$ distance pairs (intra-piece and $V\to \bnd{\rdiv}$) are stored explicitly,
computing eccentricities reduces to computing, for each pair $(P,u)$ with $u\notin V(P)$, the
quantity $e'_P(u):=\max_{v\in V(P)\setminus \bnd{P}}\{\dist_G(u,v)\}$ whenever
it is larger than the largest distance $\dist_G(u,b^*)$ from $u$ to $\bnd{P}$.
By Equation~\eqref{eq:distance-query}, $e'_P(u)>\dist_G(u,b^*)$ implies
\[e'_P(u)=\dist_G(u,b^*)+\max_{v\in V(P)\setminus \bnd{P}}\left\{\dist_P(\dirB(u,\dist_G(u,b^*))\cap V(P),v)\right\}.\]
But recall that values $\dist_P\left(\dirB(u,\dist_G(u,b))\cap V(P),v\right)$ for $b\in \bnd{P}$ and $v\in V(P)$ are computed and stored in the exact distance oracle anyway. So their maxima for
fixed boundary vertices $b$ can be found without increasing the running time asymptotically.
Hence, we conclude that all the relevant values $e'_P(u)$, and therefore the $n$ vertex eccentricities,
can be computed within $\Ot\left(n^{2-\frac{1}{3h-2}}\right)$ time.

\paragraph{Wiener index.} The Wiener index of $G$ is defined as $\sum_{u,v\in V}\dist_G(u,v)$. Suppose $G$ is strongly connected, as otherwise the Wiener index is infinite -- this can be verified in linear time.

Again, computing the Wiener index can be reduced in $O(n^2/\sqrt{r}+nr)$ time to computing
values
${\sigma_{P,u}:=\sum_{v\in V(P)\setminus \bnd{P}}\dist_G(u,v)}$ for all $P\in \rdiv$ and $u\in V\setminus V(P)$. Let again $b_1,\ldots,b_k$ denote the boundary vertices of $P$ so that
$\dist_G(u,b_1)\leq \ldots\leq \dist_G(u,b_k)$. For convenience, put $\dist_G(u,b_{k+1}):=\infty$ and $B_{i}:=\dirB(u,\dist_G(u,b_i))\cap V(P)$.
By Equation~\eqref{eq:distance-query},
we obtain:
\begin{align*}
\sigma_{P,u}&=\sum_{i=1}^k \sum_{v\in B_{i+1}\setminus B_i\setminus\bnd{P}}\left(\dist_G(u,b_i)+\dist_P\left(B_i,v\right)\right)\\
&=\sum_{i=1}^k \left(\dist_G(u,b_i) \cdot |B_{i+1}\setminus B_i\setminus \bnd{P}|+\sum_{v\in B_{i+1}\setminus B_i\setminus\bnd{P}}\dist_P\left(B_i,v\right)\right)\\
&=\sum_{i=1}^k \left(\dist_G(u,b_i) \cdot (|B_{i+1}|-|B_i|-[i\neq k\land 
 \dist_G(u,b_{i+1})>\dist_G(u,b_i)])\right)+\sum_{i=1}^k\sum_{v\in B_{i+1}\setminus B_i\setminus\bnd{P}}\dist_P\left(B_i,v\right).\\
\end{align*}
Observe that $\sigma_{P,u}$ could be computed in $O(\sqrt{r})$ time if, for 
all pairs $(B',B'')\in \pcbset_B$, the value $\sum_{v\in B'\setminus B''\setminus\bnd{P}}\dist_P\left(B'',v\right)$
was available.
Recall that $|\pcbset_P|=O(r^{h-1})$, so all these auxiliary
values could be computed in $O(|\pcbset_P|^2\cdot r)=O(r^{2h-1})$ time.
However, proceeding naively like this, we would obtain
the $\Ot\left(n^{2-\frac{1}{4h-3}}\right)$ bound that is worse than the one we have obtained for computing eccentricities.

To obtain a more efficient solution, we reuse the patterns $\delta_{P,u,i}$. Recall (\Cref{lem:ball-pattern}) that
for a piece-source pair $(P,u)$, one could compute pivots $1=\pi_1,\ldots,\pi_{\ell+1}=k+1$ such that for ${i=1,\ldots,\ell}$, the balls $\dirB(u,\dist_G(u,b_l))\cap V(P)$, where
$\pi_i\leq l<\pi_{i+1}$, were uniquely determined
by the pattern $\delta_{P,u,i}$.
For $i=1,\ldots,\ell$, consider:
\begin{align}\label{eq:wiener}
z_{P,u,i}
&=\sum_{j=\pi_i}^{\pi_{i+1}-1} \sum_{v\in B_{j+1}\setminus B_i\setminus\bnd{P}}\dist_P\left(B_j,v\right)
\end{align}
Clearly, if we can compute $\sum_{i=1}^\ell z_{P,u,i}$ in $O(\ell)$ time for a fixed pair $(P,u)$, then we are done.

Since the balls $B_{\pi_i},\ldots,B_{\pi_{i+1}-1}$ are uniquely determined by the
pattern $\delta_{P,u,i}$ (\Cref{l:interval-ball-unique}), all summands in~\eqref{eq:wiener} except the last one (for $j=\pi_{i+1}-1$) depend
on $\delta_{P,u,i}$ only. 
To make the last summand determined by $\delta_{P,u,i}$ a slight change in the
pattern's definition is enough:
let the pattern store the value \linebreak $\min(2r,\dist_G(u,b_t)-\dist_G(u,b_{\pi_i}))$
instead of $\min(r,\dist_G(u,b_t)-\dist_G(u,b_{\pi_i}))$
for $t\geq \pi_{i+1}$. This slight change does not increase the time needed to compute the patterns of their total number asymptotically.
To see that the adjustment indeed makes $B_{\pi_{i+1}}$ uniquely determined by $\delta_{P,u,i}$,
note that if $\dist_G(u,b_{\pi_{i+1}})-\dist_G(u,b_{\pi_i})<2r$, then the pattern encodes the exact value of that difference and one can proceed similarly
as in the proof of Lemma~\ref{l:interval-ball-unique}.
If, on the other hand $\dist_G(u,b_{\pi_{i+1}})-\dist_G(u,b_{\pi_i})\geq 2r$, then by
$\dist_G(u,b_{\pi_{i+1}-1})-\dist_G(u,b_{\pi_i})<r$
we conclude $\dist_G(u,b_{\pi_{i+1}})-\dist_G(u,b_{\pi_{i+1}-1})\geq r$. 
It follows that $B_{\pi_{i+1}}$ consists precisely
of the elements of 
$B_{\pi_{i+1}-1}$, the vertex $b_{\pi_{i+1}}$, and all the vertices reachable
from $\{b_1,\ldots,b_{\pi_{i+1}-1}\}$ in $P$.
In other words, $B_{\pi_{i+1}}$ is completely determined by the pattern and $P$ in such a case as well.

We have thus proved that the value $z_{P,u,i}$ can be computed based on the pattern $\delta_{P,u,i}$ and the piece~$P$ exclusively. Hence, we can store and reuse this data when processing other pairs $(P,u)$.
Finally, recall that the sets $B_{\pi_i},(B_{\pi_i+1}\setminus B_{\pi_i}),\ldots,(B_{\pi_{i+1}}\setminus B_{\pi_{i+1}-1})$,
as well as the identifiers of the respective balls $B_{\pi_i},\ldots,B_{\pi_{i+1}}$
are computed anyway when processing the pattern $\delta_{P,u,i}$ (see Section~\ref{s:speeding-balls}). As a result, the sum~\eqref{eq:wiener} can be evaluated
in $\Ot(r)$ time per encountered pattern as well, i.e., the time spent processing a single
pattern increases only by a constant factor.
We summarize the above discussion with the following.
\thmeccentricitieswiener*

\section{Decremental reachability and bottleneck distances}\label{s:decremental}
In this section, we show how the bounded VC dimension
of the ball set in $G$ can be used to obtain a
non-trivial (partially) dynamic reachability data structure with $O(1)$ query time.

More specifically, we consider the decremental setting.
At the beginning, a $K_h$-minor-free digraph $G$ is given.
The graph $G$ is subject to edge deletions.
The goal is to obtain a data structure that can process any sequence of (at most $m$) edge deletions
in $G$ possibly interleaved with arbitrary-pair reachability queries in the current version of $G$.

We will first show a simpler randomized data structure. Later, we will discuss how to achieve the same asymptotic bounds deterministically. We will use the following result.
\begin{theorem}\label{thm:dec-ssr}{\upshape\cite{BernsteinPW19}}
    Let $G=(V,E)$ be a directed graph and let $s\in V$ be a fixed source. There exists a data structure maintaining the set of vertices reachable from $s$ in $G$ when $G$ is subject to edge deletions.
    Specifically, after each deletion, the data structure reports the vertices that became unreachable from $s$ as a result of that deletion.
    The expected total update time of the data structure is $\Ot(n+m)$.
\end{theorem}

\subsection{The maintained components}
The base of our data structure is an $r$-division $\rdiv$ of the initial graph. Note that $\rdiv$ -- set up once -- remains a valid $r$-division of $G$ even when $G$ is subject to deletions. Specifically, whenever an edge $e\in E$ is deleted, it is deleted from the corresponding piece $P\in\rdiv$ it is contained in. Clearly, the properties of the $r$-division (see~\Cref{sec:template}) are preserved.
Recall that computing $\rdiv$ costs $O(n^{1+\eps})$ time for $h=O(1)$.

\newcommand{\dssr}{\mathcal{L}}

\paragraph{Intra-piece decremental reachability.}
For each $P\in \rdiv$ and $u\in V(P)$ we set up a
decremental
single-source reachability (SSR)  data structure
of~\Cref{thm:dec-ssr} on the graph $P$ with source $u$. As each initial edge~$e$ is contained in precisely one piece of $\rdiv$, the deletion of $e$ is passed to the corresponding piece. These data structures allow checking for intra-piece paths in $O(1)$ time. Their total update time is $\Ot((n/r)\cdot r\cdot r)=\Ot(nr)$.

\paragraph{Decremental reachability to boundary vertices.} For each $b\in \bnd{\rdiv}$, we set up a decremental SSR  data structure
$\dssr_b$ of~\Cref{thm:dec-ssr} on the \emph{reverse} $G^R$ of $G$  with source $b$.
Note that the data structures $\dssr_b$, $b\in \bnd{\rdiv}$, can be used to explicitly maintain, for all pairs $(P,u)\in \rdiv\times V$, the subset $R_{P,u}$ of boundary vertices $\bnd{P}$ currently reachable from $u$ in $G$.
Indeed, whenever some $u\in V$ is reported as unreachable from some $b\in \bnd{\rdiv}$ in $G^R$, this means that $b$ becomes unreachable from $u$ in $G$. In such a case, $b$ can be removed from all $R_{P,u}$ such that $b\in \bnd{P}$.
Each deletion issued to $G$ is passed to each of the
$\Ot(n/\sqrt{r})$ data structures $\dssr_b$. Note that their total update time is therefore $\Ot(n^2/\sqrt{r})$. Since $\sum_{P\in\rdiv}|\bnd{P}|=O(n/\sqrt{r})$, the sets
$R_{P,u}$ undergo $O(n^2/\sqrt{r})$ updates in total.

\paragraph{Pattern identifiers.} Moreover, for every piece $P\in\rdiv$, we set up a dynamic strings data structure $\tcoll_P$ (see~\Cref{t:dyn-string}) maintaining the subsets $R_{P,u}\subseteq \bnd{P}$ encountered so far (in any of the versions of the graph~$G$), so that they are assigned unique small-integer identifiers. Recall that such subsets can be viewed as strings
from~$\{0,1\}^{|\bnd{P}|}$. We will also call
these subsets \emph{patterns}.
Whenever a vertex $b\in \bnd{P}$ gets removed
from some~$R_{P,u}$, this corresponds to a single-coordinate update in the string representing it,
so a new identifier of the updated set
$R_{P,u}\in\tcoll_P$ can be obtained in $\Ot(1)$ time. Observe that this way, the data structure maintains small-integer identifiers of all the $O(n^2/r)$ sets $R_{P,u}$ (for every individual $P\in\rdiv$) subject to a sequence of edge deletions in $G$. The total time spent maintaining the patterns is $\Ot(n^2/\sqrt{r})$.

\paragraph{Decremental reachability from patterns.} Finally, for a fixed $P$, and every encountered pattern $X\subseteq \bnd{P}$ (of the form $R_{P,u}$ for some $u$), we set up a decremental SSR data structure $\dssr_{P,X}$ on the (current version of the) piece $P$ with a super-source $s$
and auxiliary edges $sb$ for all $b\in X$ added. Every subsequent edge deletion that $P$ undergoes is passed to $\dssr_{P,X}$, so that $\dssr_{P,X}$
always maintains the vertices reachable from the set $X$ in $P$. By~\Cref{thm:dec-ssr}, the
total update time of $\dssr_{P,X}$ is $\Ot(r)$.
As a result, the total cost of maintaining
these decremental SSR data structures can
be bounded by $\sum_{P\in\rdiv}\Ot\left(|\tcoll_P|r\right)$,
where $|\tcoll_P|$ is the size of the collection $|\tcoll_P|$ at the end of the sequence of updates, or equivalently, the total number of distinct patterns encountered for the piece $P$.

\subsection{Answering queries} Suppose a query $(u,v)$ is issued whether $u$ can currently reach $v$ in $G$. Let $P$ be any piece containing $v$.
If $u\in V(P)$, a path $u\to v$ contained entirely in the piece $P$ may exist. We can test that using the intra-piece decremental SSR data structure.

If $u\notin V(P)$ or no $u\to v$ path inside $P$ exists, it is enough to check whether~$v$ is reachable from
$R_{P,u}$ in the piece $P$ -- recall that this is maintained by the data structure $\dssr_{P,R_{P,u}}$.
Indeed, if $v\in \bnd{P}$, then $u$ can reach $v$ iff $v\in R_{P,u}$ by the definition.
Otherwise, $P$ is the unique piece containing
$v$, and any $u\to v$ path in $G$ contains some latest $b\in \bnd{P}$. So an $u\to v$ path
exists in $G$ iff $v$ is reachable in $P$ from some of its boundary vertices reachable from $u$ in $G$. 
Clearly, such a query can be processed in constant time as well.

\subsection{Bounding the total update time}\label{sec:decr-total-update}
The expected total update time can be easily seen to be $\Ot(n^2/\sqrt{r}+nr)$ plus $\Ot(r)$ times the
total number of distinct patterns encountered through all pieces. In the following, we will prove the below key lemma.
\begin{lemma}\label{l:reach-patterns}
    For any piece $P\in\rdiv$, the total number of distinct patterns encountered is $O(r^{(h-1)/2})$.
\end{lemma}
By~\Cref{l:reach-patterns}, it will follow that the expected total update time
of the data structure is $\Ot\left(n^2/\sqrt{r}+nr^{(h-1)/2}\right)$. This  is optimized for $r=n^{2/h}$, and leads to the $\Ot\left(n^{2-1/h}\right)$ bound then.

First, let us consider the concept of \emph{bottleneck distances} in a (static) weighted graph $G=(V,E)$ with edge weights given by $w$.
The \emph{bottleneck length} of a path $e_1\ldots e_k=P\subseteq G$ is defined as $\max_{i=1}^k\{w(e_i)\}$.
For $s,t\in V$, the bottleneck distance $\beta_G(s,t)$
is defined as a minimum bottleneck length of an $s\to t$ path in~$G$. For any $s\in V$ and $r\in\mathbb{R}$, define the bottleneck ball
$\dirB^\beta(s,\rho)=\{v\in V:\beta_G(s,v)\leq \rho\}$.
\begin{lemma}\label{l:bottleneck-vc}
If a weighted digraph $G=(V,E)$ with weights given by $w$ is $K_h$-minor-free, then the set system
$\dirBs^\beta=\left\{\dirB^\beta_G(s,\rho):s\in V, \rho\in\mathbb{R}\right\}$ has VC dimension at most $h-1$.
\end{lemma}
\begin{proof}
    Consider a graph $G'$ obtained from $G$ by changing its edge weights:
    the weight of $e\in E$ is given by $n^j$ in $G'$ if there are $j$ edges in $G$ less than or equal $w(e)$.

    We now prove that the set system
    $\mathcal{Q}=\left\{\dirB_{G'}(s,n^{i+1}-1):s\in V, i\in [m]\right\}$ contains $\dirBs^\beta$.
    Let ${\dirB^\beta(s,\rho)\in \dirBs^\beta}$.
    Note that all bottleneck lengths come from
    the set of edge weights of $E$, so
    in fact $\dirB^\beta(s,\rho)=\dirB^\beta_G(s,w(e))$, for some $e\in E$.
    Suppose there are precisely $j$ edges of weight at most~$w(e)$ in $G$.
    We now argue that $\dirB^\beta(s,w(e))=\dirB_{G'}\left(s,n^{j+1}-1\right)$.
    Indeed, if for some $v\in V$, $\beta_G(s,v)\leq w(e)$, then there exists a simple path $\pi=s\to v$ whose all weights are at most $w(e)$. So the (standard) weight of $\pi$ in~$G'$
    is at most $(n-1)\cdot n^j=n^{j+1}-n\leq n^{j+1}-1$.
    This proves that $v\in \dirB_{G'}(s,n^{j+1}-1)$.
    Now, if there exists a simple path $\pi'=s\to v$ in $G'$ of weight at most $n^{j+1}-1$, then every individual edge length in $\pi'$ is at most
    $n^j$, since all edge weights in $G'$ are powers of $n$.
    So all the edges of $\pi'$ have weights no more than $w(e)$ in~$G$, which proves $v\in \dirB^\beta(s,w(e))$.

    Since $\mathcal{Q}\subseteq \dirBs_G$, we also have $\dirBs^\beta\subseteq\dirBs_G$. As $\dirBs_G$ has VC dimension at most $h-1$ by~\Cref{cor:vcdim}, so does $\dirBs^\beta$.
\end{proof}
We are now ready to prove~\Cref{l:reach-patterns}.
\begin{proof}[Proof of~\Cref{l:reach-patterns}]
Without loss of generality, suppose that all $m$ edges of $G=(V,E)$ are eventually deleted. 
Let $G_0,G_1,\ldots,G_m$ be the subsequent versions
 of the graph $G$ given in reverse order, so that
 $E(G_0)=\emptyset$ and $|E(G_m)|=m$.
For each edge $e\in E$, let $w(e)$ denote the smallest index $i$ such that $G_{i}$ contains
the edge $e$. Clearly, we have $e\in G_{j}$ if and only if $j\geq w(e)$. 

Let $G'$ be the (initial) graph $G$ with weights given by $w(\cdot)$. We now prove that for any $s,t\in V$, and $k\in \{0,\ldots,m\}$, a path
$s\to t$ exists in $G_k$ if and only if
$\beta_{G'}(s,t)\leq k$. Indeed, if $\pi$ is
an $s\to t$ path in $G_k$, then each edge $e$ of $\pi$ satisfies $w(e)\leq k$. It follows that $\pi$ has bottleneck length at most $k$ in $G'$.
Conversely, if $\pi'$ is an $s\to t$ path in $G'$
of bottleneck length $k$, then all edges of $\pi'$
have we weight at most $k$ in $G'$, so they are
all contained in $G_k$.

Now, fix some piece $P\in \rdiv$, and suppose a pattern $R_{P,u}$ appears in the decremental reachability data structure.
This means that for some $j$, $R_{P,u}$ is the
set of vertices from $\bnd{P}$ reachable
from $u$ in $G_j$.
Equivalently, for $b\in\bnd{P}$ we have
$b\in R_{P,u}$ if and only if
$\beta_{G'}(u,b)\leq j$.
This means that $R_{P,u}=\dirB^\beta_{G'}(u,j)\cap \bnd{P}$.
So $R_{P,u}$ belongs to the restriction
of the set system $\dirBs^\beta_{G'}$
to $\bnd{P}$. By~\Cref{l:bottleneck-vc} and~\Cref{cor:restrict}, the size
of this set system is $O(|\bnd{P}|^{h-1})=O(r^{(h-1)/2})$. The lemma follows.
\end{proof}

\subsection{A deterministic data structure}
Note that the only source of randomness so far was the decremental SSR data structure of~\Cref{thm:dec-ssr}. So efficient \emph{deterministic} data structures for decremental SSR have not been described yet, even for minor-free graphs.
However, in our application, many instances
of decremental SSR are related. We will exploit that.

We use the following result is proved in~\cite[Section~3.1]{Karczmarz18}.\footnote{There, the desired graph $G^*$ is denoted $\mathcal{R}^+(G)$. \cite[Lemma~3.2]{Karczmarz18} proves the desired properties of that graph. That this graph is maintained deterministically in time $\Ot(n^{3/2})$ is proved in~\cite[Lemma~3.6]{Karczmarz18} and the discussion afterward.}
\begin{theorem}\label{thm:dec-emulator}
    Let $G=(V,E)$ be a $K_h$-minor-free digraph (where $h=O(1)$) that undergoes edge deletions. There exists a \emph{deterministic} data structure maintaining a related subgraph $G^*$ on $V$ that is also decremental, has $\Ot(n)$ edges,
    and satisfies the following at all times:
    \begin{itemize}
        \item for all $s,t\in V$, $s$ can reach $t$ in $G$ if and only if $s$ can reach $t$ in $G^*$,
        \item for all $s,t\in V$ such that $s$ can reach $t$ in $G^*$, $\dist_{G^*}(s,t)=O(\log{n})$.
    \end{itemize}
    The total update time of the data structure is $\Ot\left(n\sqrt{n}\right)$.
\end{theorem}

We will also need the following classical result, the so-called \emph{Even-Shiloach tree}.
\begin{theorem}\label{thm:es-tree}{\upshape\cite{EvenS81, HenzingerK95}}
Let $G=(V, E)$ be a digraph with a fixed source $s\in V$. Then, the set $\{v\in V:\dist_G(s,v)\leq d\}$ can be maintained subject to edge deletions in $G$ in $O(d(n+m))$ total update time.
\end{theorem}
\begin{corollary}\label{cor:det-dec-ssr}
Let $G=(V,E)$ be an $n$-vertex $K_h$-minor-free digraph and suppose a graph $G^*$ of~\Cref{thm:dec-emulator} is maintained for $G$. Then, given some $S\subseteq V$, one can set up a deterministic data structure maintaining vertices reachable from $S$ in $G$ whose total update time is $\Ot(n)$.
\end{corollary}
\begin{proof}
    It is enough to construct the data structure of~\Cref{thm:es-tree} with $d=\Theta(\log{n})$ for a graph $G^*_S$ obtained from $G^*$ by augmenting it with a super-source $s$ with outgoing edges to the vertices of $S$. Note that the size of $G^*_S$
    is $\Ot(n)$, and thus the total update time
    of the data structure is $\Ot(n)$. All the edge deletions that $G^*$ undergoes are passed to the data structure maintaining $G^*_S$.
    Clearly, every vertex reachable from $S$ in $G$
    is reachable from $s$  in $G^*_S$ via a path of length $O(\log{n})$.
\end{proof}
To obtain a deterministic decremental all-pairs reachability data structure, we apply~\Cref{thm:dec-emulator} to the (reverse) input graph $G$, and also separately to all the
individual pieces $P\in\rdiv$. Note that the total initialization and update cost incurred is $\Ot(n^{3/2}+(n/r)r^{3/2})=\Ot(n^{3/2})$.
Observe that this allows us to replace every instance
of the randomized decremental SSR data structure of~\Cref{thm:dec-ssr} with a respective data structure of~\Cref{cor:det-dec-ssr}. This
comes at no asymptotic cost wrt. $\Ot$ notation.
We have thus proved:
\decreachability*

\subsection{The bottleneck distance oracle}\label{sec:bottleneck}
Finally, we remark that \Cref{thm:dec-reachability} almost immediately implies a bottleneck distance oracle with the same preprocessing time.
\begin{lemma}\label{l:bottleneck-oracle}
    Let $G$ be a weighted $K_h$-minor-free directed graph. An $O(n^{2-1/h})$-space bottleneck distance oracle for $G$ with $O(\log{n})$ query time can be constructed deterministically in $\Ot(n^{2-1/h})$ time.
\end{lemma}
\begin{proof}[Proof sketch.]
    Set up a data structure of~\Cref{thm:dec-reachability} on the (unweighted) graph $G$.
    Let $E=\{e_1,\ldots,e_m\}$ so that
    $w(e_1)\leq \ldots \leq w(e_m)$.
    Remove all edges of $G$ from the data structure one by one in the order
    of decreasing weights: $e_{m},e_{m-1},\ldots,e_1$.
    Let $G_m,G_{m-1},\ldots,G_0$ be the subsequent obtained versions of $G$.
    Analogously as in~\Cref{l:reach-patterns}, one
    can easily argue that if $k$ is the smallest index
    such that a path $s\to t$ exists in $G_k$,
    then $\beta_G(s,t)=w(e_k)$.
    Given a query $(s,t)$, one could thus
    binary search for the value $k$ and use
    that for computing the bottleneck distance.
    Since the data structure of~\Cref{thm:dec-reachability} has constant query time, this would imply an $O(\log{n}\cdot \log\log{n})$ query time bound if we made the data structure persistent in a black-box way, e.g., using~\cite{Dietz89a}.
    
    That being said, the $\log\log{n}$ persistence overhead can be avoided rather easily. It is enough to augment the decremental data structure as follows.
    First, for each pair $(P,s)$, store timestamped
    pointers to the data structures $\dssr_{P,R_{P,s}}$ for all the versions 
    of the set $R_{P,s}$. Second, for each data structure $\dssr_{P,X}$, store the timestamps of when individual vertices of $P$ stopped being reachable from $X$.
    Using this additional data, given a query $(s,t)$ such that $t\in V(P)$, one can perform the binary search on the list of pointers for $(P,s)$, so that decisions are made in $O(1)$ time based on the timestamps indicating when $t$ became unreachable in the respective data structures $\dssr_{P,R_{P,s}}$. By proceeding similarly with the intra-piece decremental SSR data structures, we can find intra-piece bottleneck distances in $O(1)$ time.

    The space usage is trivially bounded by the construction time, i.e., $\Ot(n^{2-1/h})$, but it is easy to see that the amount of data needed for performing the queries is actually $O(n^{2-1/h})$.
\end{proof}

\section{Approximate distance oracles}\label{sec:approx}

In this section, we show more efficient and conceptually simpler (compared to the exact oracles of Sections~\ref{sec:unweighted}~and~\ref{sec:weighted}) $(1+\eps)$-\emph{approximate} distance oracles for non-negatively weighted minor-free graphs with polylogarithmic query time. 
We consider two variants of the problem: (1) the bounded aspect ratio case when the edge weights come from the set $\{0\}\cup [1,W]$, and (2) the general case with arbitrary real non-negative weights.
We will assume that the accuracy parameter satisfies $\eps\in (0,1)$ and $1/\eps=\poly{n}$.

\subsection{Bounded aspect ratio}\label{sec:bounded-ratio}
Let $\rdiv$ be an $r$-division of $G$. Similarly as we did in~\Cref{sec:unweighted}, we first compute the
exact distances $\dist_G(u,v)$ for all pairs $(u,v)$
such that either (1) $(u,v)\in V\times\bnd{\rdiv}$,
or (2) $(u,v)\in V(P)$ for some $P\in \rdiv$. Recall that this can be done in $\Ot(n^2/\sqrt{r}+nr)$ time naively.
However, eventually, only the exact distances of the latter type will be actually stored by the data structure; this allows taking into account shortest paths that do not pass through $\bnd{P}$. 

Let $\ell$ be the maximum integer such that $(1+\eps)^\ell\leq 2nW$.
Note that $\ell=O\left(\log_{1+\eps}(nW)\right)=O(\log(nW)/\eps)$.

For each pair $(P,u)\in \rdiv\times V$, and each integer $j=0,\ldots,\ell$, let us define
\[X_{P,u,j}=\{b\in \bnd{P}:\dist_G(u,b)\leq (1+\eps)^j\}=\dirB_G\left(u,(1+\eps)^j\right)\cap \bnd{P}.\]
Moreover, define $X_{P,u,-1}:=\dirB_G(u,0)\cap \bnd{P}$.
By~\Cref{cor:vcdim}~and~\Cref{cor:restrict}, the number of distinct sets of the form $X_{P,u,j}$ for a fixed $P$ is $O\left(|\bnd{P}|^{h-1}\right)=O\left(r^{(h-1)/2}\right)$. The data structure
stores, for each such $X\subseteq\bnd{P}$, the distances $\dist_P(X,v)$ for all $v\in V(P)$.
Note that these distances can be computed in $\Ot(r)$ time and stored in $O(r)$ space per subset
$X$. As a result, through all pieces $P$ and sets $X$, the space cost incurred is $O(nr^{(h-1)/2})$.

Additionally, for each $(P,u)\in \rdiv\times V$, the data structure stores the $\ell+2=O\left(\log(nW)/\eps\right)$ pointers to the sets $X_{P,u,-1},X_{P,u,0},\ldots,X_{P,u,\ell}$ as defined above, so that the distances
$\dist_P(X_{P,u,j},v)$ can be accessed in constant time. Note that this costs $O((n^2/r)\cdot \log(nW)/\eps)$ space.

Summing up, the data structure uses
$O((n^2/r)\log(nW)/\eps+nr^{(h-1)/2})$ space.

\paragraph{Answering queries.} Recall that the (exact) answers for queries involving two vertices of a single piece are stored in the data structure explicitly. Let us thus focus on queries $(u,v)$ such
that $v\in V(P)$ and $u\notin V(P)$ for some $P\in\rdiv$. To answer such approximate distance queries we simply output the estimate:\[d'(u,v):=\min\left(\dist_P(X_{P,u,-1},v),\min_{j=0}^\ell\left\{(1+\eps)^j+\dist_P(X_{P,u,j},v)\right\}\right).\]
The above
can be evaluated in $O(\ell)=O(\log(nW)/\eps)$ time.
We now prove this query procedure correct.
\begin{lemma}\label{l:apx-correct}
    Let $u,v$ be such that $v\in V(P)$ and $u\notin V(P)$ for some piece $P\in\rdiv$. Then, we have:
    \[\dist_G(u,v)\leq \dist'(u,v)\leq (1+\eps)\cdot\dist_G(u,v).\]
\end{lemma}
\begin{proof}
    Since every $u\to v$ path has to contain some last boundary vertex of $P$, we have:
    \begin{align*}
    \dist'(u,v)&\geq \min_{j=-1}^\ell\left\{\max_{b\in X_{P,u,j}}\{\dist_G(u,b)\}+\min_{b\in X_{P,u,j}}\{\dist_P(b,v)\}\right\}\\
    &\geq \min_{j=-1}^\ell\min_{b\in X_{P,u,j}}\{\dist_G(u,b)+\dist_P(b,v)\}\\
    &\geq \min_{b\in \bnd{P}}\{\dist_G(u,b)+\dist_P(b,v)\}=\dist_G(u,v).
    \end{align*}
    This proves the lower bound on $\dist'(u,v)$ and
    also the upper bound if $\dist_G(u,v)=\infty$.

    Suppose $\dist_G(u,v)<\infty$ and let $\pi$ be the shortest $u\to v$ path in $G$.
    Once again, $\pi$ contains some latest boundary vertex $b^*\in \bnd{P}$, and $\dist_G(u,b^*)<nW$.
    As a result, we have $b^*\in X_{P,u,i}$ for some smallest $i$. Note that by the minimality of $i$,
    if $i\geq 0$, then
    we have $(1+\eps)^i\leq (1+\eps)\cdot \dist_G(u,b^*)$. Hence, we obtain:
    \[ d'(u,v)\leq (1+\eps)^i+\dist_P(X_{P,u,i},v)\leq (1+\eps)\cdot \dist_G(u,b^*)+\dist_P(b^*,v)\leq(1+\eps)\cdot \dist_G(u,v).\]
    The special case when the minimal $i$ equals $-1$,
    i.e. $\dist_G(u,b^*)=0$, is similar.
\end{proof}

\paragraph{Construction time and an optimization.} If one constructs all the sets $X_{P,u,i}$ explicitly and naively,
the construction time is $\Ot((n^2/\sqrt{r})\log(W)/\eps+nr^{(h-1)/2})$. One can remove the $\log(W)/\eps$ dependence
from the first term by using e.g., the dynamic strings data structure (\Cref{t:dyn-string}) for representing these
sets for each piece $P$. Indeed,
note that we have $X_{P,u,0}\subseteq X_{P,u,1}\subseteq \ldots \subseteq X_{P,u,\ell}$, 
so the identifiers of all the \emph{distinct} out of these sets can be obtained by performing $|\bnd{P}|=O(\sqrt{r})$ dynamic strings operations.
The dynamic strings data structure allows establishing the needed pointers for $(P,u)$ in $\Ot(\sqrt{r}+\ell)$ time as well. 

We can make further use of the idea of pruning the repeated sets $X_{P,u,i}$  equal to the preceding set $X_{P,u,i-1}$. Indeed, if one sets
\[I_{P,u}=\{j=0,\ldots,\ell:X_{P,u,j}\neq X_{P,u,j-1}\},\]
then the proof of~\Cref{l:apx-correct} goes through
even with a slightly relaxed estimate formula \[d'(u,v):=\min\left(\dist_P(X_{P,u,-1},v),\min_{j\in I_{P,u}}\left\{(1+\eps)^j+\dist_P(X_{P,u,j},v)\right\}\right).\]
Moreover, $|I_{P,u}|=\min(\ell+1,|\bnd{P}|)$,
so the pointers for the unique sets $X_{P,u,i}$ can be found in $\Ot\left(\min(\sqrt{r},\ell)\right)$ time. This optimization reduces the construction
time to $\Ot(n^2/\sqrt{r}+nr^{(h-1)/2})$
and at the same time the space usage to $O\left(\min\left((n^2/r)\log(nW)/\eps,n^2/\sqrt{r}\right)+nr^{(h-1)/2}\right)$.
By setting $r=n^{2/(h+1)}$, we obtain the following for reasonably small values of $W$.
\begin{theorem}
Let $G$ be a non-negatively weighted $K_h$-minor-free directed graph whose positive weights come from the interval $[1,W]$. Let $\eps\in [1/\poly(n),1]$.
There exists an $O(n^{2-2/(h+1)}\log(nW)/\eps)$-space data structure supporting arbitrary-pair $(1+\eps)$-approximate distance queries in $G$ in $O(\log(nW)/\eps)$ time. The data structure can be constructed in $\Ot\left(n^{2-1/(h+1)}\right)$ time.
\end{theorem}

\subsection{The unbounded case}
Note that the space of the (optimized) bounded aspect ratio data structure could also be bounded by $O\left(n^2/\sqrt{r}+nr^{(h-1)/2}\right)$ regardless of the value of $W$. The preprocessing could be bounded using the same expression, modulo the polylogarithmic factors. In this section, we show how to optimize the query time -- within this space and preprocessing time budget -- to $O(\log(n)/\eps)$, i.e., independent of $W$.

Note that by dividing all the edge weights by $c\cdot \eps$, where $c$ is the minimum positive weight, we can assume that all edge weights are either zero or at least $1/\eps\geq 1$. When returning an estimate, we will scale the estimate back by $c\cdot \eps$.

Suppose at query time (for $u\notin V(P)$ and $v\in V(P)$), one is given a crude estimate interval $[x,y]$ satisfying $\dist_G(u,v)\in [x,y]$, where $x\geq 1/\eps$. Let $p\geq 0$ be the largest integer satisfying $(1+\eps)^{p}\leq \eps\cdot x$,
and let $q$ be the smallest integer such that $(1+\eps)^{q}\geq y$. Then, if
we further relax the estimate formula to the following:
\[d'(u,v):=\min\left(\dist_P(X_{P,u,-1},v),\min_{j\in I_{P,u}\cap [p,q]}\left\{(1+\eps)^j+\dist_P(X_{P,u,j},v)\right\}\right),\]
$\dist'(u,v)$ remains a valid $(1+\eps)$-approximate estimate of $\dist_G(u,v)$. Indeed, since we minimize over a smaller subset of indices $j$ than in~\Cref{l:apx-correct}, 
the lower bound $d'(u,v)\geq \dist_G(u,v)$ is preserved.
For the upper bound, let us recall the argument from~\Cref{l:apx-correct} for the case $\dist_G(u,v)>0$ (the zero-distance case is handled exactly as before). Let $b^*\in \bnd{P}$ be the latest vertex of $\bnd{P}$ on some shortest $u\to v$ path in $G$. Then, there exists some smallest $i$ such that $b^*\in X_{P,u,i}$, and $(1+\eps)^i\leq (1+\eps)\cdot \dist_G(u,b^*)$. If $i\in [p,q]$, 
then the argument from~\Cref{l:apx-correct} goes through. Note that $i>q$ cannot happen,
since then we would have $\dist_G(u,b^*)>(1+\eps)^{q}\geq y\geq \dist_G(u,v)$.
On the other hand, if $i<p$, then $b^*\in X_{P,u,p}$ and
\[
d'(u,v)\leq (1+\eps)^p+\dist_P(X_{P,u,p},v)\leq \eps\cdot x+\dist_P(b^*,v)\leq \eps\cdot \dist_G(u,v)+\dist_G(u,v)=(1+\eps)\dist_G(u,v).
\]
Note that computing the relaxed estimate $d'(u,v)$ can be performed
in $O(\log{n}+\log(y/(\eps x))/\eps)=O(\log(ny/x)/\eps)$ time since
$|I_{P,u}\cap [p,q]|=O(\log_{1+\eps}(y/(\eps x)))$
and locating the boundary elements of $I_{P,u}\cap [p,q]$
can be performed via binary search in $O(\log{|I_{P,u}|})=O(\log{r})$ time given the interval $[x,y]$.

Observe that if we could efficiently find an estimate interval $[x,y]$ such that $x\geq 1/\eps$ and $y/x=\poly{n}$, then the query procedure would run in $O(\log(n)/\eps)$ time.
The following observation shows that using bottleneck distances to produce such rough estimates works well.
\begin{observation}
For any $u,v\in V$, we have $\beta_G(u,v)\leq \dist_G(u,v)< n\cdot \beta_G(u,v)$.
\end{observation}
\begin{proof}
This is trivial if $\dist_G(u,v)=\infty$.
Otherwise, if $Q$ is a shortest $u\to v$ path in $G$, then clearly
all edges on $Q$ are at most $\dist_G(u,v)$, so indeed $\beta_G(u,v)\leq \dist_G(u,v)$.
If $Q'$ is the shortest $u\to v$ simple path in terms of the bottleneck distance, then all edges
on $Q'$ have weight at most $\beta_G(u,v)$. So the standard weight of $Q'$ is at most
$(n-1)\cdot \beta_G(u,v)$ which implies the desired upper bound on $\dist_G(u,v)$.
\end{proof}
The above observation implies that the estimate interval $[x,y]$ such that $y/x\leq n$
and $\dist_G(u,v)\in [x,y]$ can be found using a bottleneck distance oracle.
In particular, if $x=0$ then $y=0$ and we can conclude that $\dist_G(u,v)=0$. Otherwise, we have $x\geq 1/\eps$ as desired.
Recall from~Sections~\ref{sec:decr-total-update}~and~\ref{sec:bottleneck} that a bottleneck distance oracle with
$O(n^2/\sqrt{r}+nr^{(h-1)/2})$
space and $O(\log{n})$ query time can be constructed in $\Ot(n^2/\sqrt{r}+nr^{(h-1)/2})$ time.
Note that these space and construction time bounds match those of~\Cref{sec:bounded-ratio} for the case of unbounded~$W$. We have thus proved the following theorem.

\begin{theorem}
Let $G$ be a $K_h$-minor-free directed graph with non-negative real edge weights. Let $\eps\in [1/\poly(n),1]$.
There exists an $O(n^{2-1/h})$-space data structure supporting arbitrary-pair $(1+\eps)$-approximate distance queries in $G$ in $O(\log(n)/\eps)$ time. The data structure can be constructed in $\Ot(n^{2-1/h})$ time.
\end{theorem}

\clearpage

\bibliographystyle{alphaurl}
\bibliography{ref}

\pagebreak
\appendix

\section{Proof of the minor building lemma}
\label{ap:minor_building}

The proof in this section for graphs where shortest paths are unique is (implicitly) presented in \cite{LeW24}; we reproduce it here for completeness. 

\begin{restatable}[Minor building via unique shortest paths~\cite{LeW24}]{lemma}{minorbuildingunique}
\label{lem:minor_building_unique}
Let $G = (V, E)$ be a graph where all shortest paths are unique.
Suppose we have a set of vertices $v_1, ..., v_d$ such that for every $1\le i < j \le d$
there is a vertex $t_{ij}$ that satisfies the following:
\begin{enumerate}
    \item[(1)]  For all distinct $i,j,p,q\in [d]$, $\pi(t_{ij} , v_i)$ and $\pi(t_{pq} , v_p)$ are vertex disjoint.
    \item[(2)]  For all distinct $i,j,p\in [d]$, $V(\pi(t_{ij} , v_i)) \cap V(\pi(t_{jp} , v_j)) \subseteq \{t_{ij}\}$.
    \item[(3)] For all distinct $i, j \in [d]$, $\pi(t_{ij} , v_i)$ and $\pi(t_{ij} , v_j)$ are internally vertex disjoint. 
\end{enumerate}
\text{Then $G$ contains $K_d$ as a minor.} 
\end{restatable}

For convenience, we will define $t_{ji} = t_{ij}$ for $j>i$.
We define $\Bunch(v_i)$ to be a collection of edges consisting of the undirected versions of the edges of $\bigcup_j \pi(t_{ij}, v_i)$. 
As all shortest paths are unique, $\Bunch(v_i)$ is a tree.
We say $t_{ij}$ is an \emph{endpoint} of $\Bunch(v_i)$ if $t_{ij}$ is a leaf in the tree of $\Bunch(v_i)$.

We first prove the following lemma.
\begin{lemma} \label{lem:end_points}
    If conditions (1), and (2) of \Cref{lem:minor_building_unique} hold, then $t_{ip}$ is an endpoint of $\Bunch(v_i)$, $\Bunch(v_p)$, or both.
\end{lemma}
\begin{proof}
Suppose to the contrary that $t_{ip}$ is neither an endpoint of $\Bunch(v_i)$ nor $\Bunch(v_p)$.
Then there exists an endpoint $t_{ij}$ such that the path $\pi(t_{ij} , v_i)$ passes through $t_{ip}$.
Similarly, there must exist an endpoint $t_{pq}$ such that the path $\pi(t_{pq} , v_p)$ passes through $t_{ip}$.
This means that $V(\pi(t_{ij}, v_i)) \cap  V(\pi(t_{pq}, v_p)) \supseteq \{t_{ip}\}$.
If $j\neq q$ this contradicts condition (1), and otherwise when $j=q$ this contradicts (2).
\end{proof}

\begin{figure}
    \centering
    \begin{subfigure}{0.45\textwidth}
    \includegraphics[page=1]{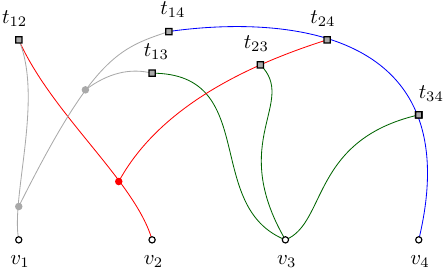}
    \caption{Example $H$ with bunches in color.}
    \label{fig:minor_building_proof1}
    \end{subfigure}
    \qquad
    \begin{subfigure}{0.45\textwidth}
    \includegraphics[page=2]{graphics/minor_building_proof.pdf}
    \caption{$H'$ obtained after contraction of $H$.}
    \label{fig:minor_building_proof2}
    \end{subfigure}
    \caption{Illustrations for proof of \Cref{lem:minor_building_unique}.}
\end{figure}

Now we present the proof of \Cref{lem:minor_building_unique}.
\begin{proof}[Proof of \Cref{lem:minor_building_unique}]
Let $H$ denote the graph induced by $\bigcup_i \Bunch(v_i)$ (see \Cref{fig:minor_building_proof1}). We show that $H$ contains $K_d$ as a minor. 

In $H$, for every $v_i$, contract all vertices that are not endpoints of $\Bunch(v_i)$ to $v_i$.
Call this resulting graph $H'$ (see \Cref{fig:minor_building_proof2}).
By \Cref{lem:end_points}, for every $i<j$, we have that either $v_i$ was not contracted with $t_{ij}$, or $v_j$ was not contracted with $t_{ij}$, or both.
Thus in $H'$, between every $v_i$ and $v_j$, we either have a direct edge if $t_{ij}$ was contracted, or a path of two edges with $t_{ij}$ as the intermediate vertex (by condition (3)).
Contracting all such paths of length $2$ in $H'$ results in $K_d$.
\end{proof}

Now we are ready to present the proof of \Cref{lem:minor_building} that we restate here for convenience.

\minorbuilding*

Notice there are two major differences between \Cref{lem:minor_building} and \Cref{lem:minor_building_unique}. The first is that \Cref{lem:minor_building} does not have the condition (3) of \Cref{lem:minor_building_unique}. 
The second is that we don't assume the uniqueness of shortest paths, and require that (1) and (2) hold for all shortest paths. Our proof will deal with both of these differences, handling the first with the Isolation Lemma and the second by careful contractions.
\begin{proof}[Proof of \Cref{lem:minor_building}]
Let $\Ghat$ be the graph obtained from $G$ by adding a tiny perturbation of the weights that guarantees all shortest paths are unique, and all shortest paths in $\Ghat$ are shortest paths in $G$. The existence of $\Ghat$ is given by the Isolation Lemma of \cite{ValiantV86}. 
In $\Ghat$, the \Cref{lem:end_points} holds for $\Ghat$ if we define $\Bunch(v_i)$ to be the undirected versions of the edges of $\bigcup_j \pi_{\Ghat}(t_{ij}, v_i)$. However, as we pointed out earlier, we no longer have the condition (3) of \Cref{lem:minor_building_unique}, 
so the paths $\pi_{\Ghat}(t_{ij},v_i)$ and $\pi_{\Ghat}(t_{ij},v_j)$ might not be vertex disjoint.
Fortunately, as these paths are shortest paths and the shortest paths in $\Ghat$ are unique, only a prefix on the path 
$\pi_{\Ghat}(t_{ij}, v_i)$
can be shared with the path
$\pi_{\Ghat}(t_{ij}, v_j)$.

This suggests a simple modification of the proof of \Cref{lem:minor_building_unique} where we looked for a minor in the graph $H$ induced by $\bigcup_{i}\Bunch(v_i)$.
For each $i,j\in [d]$, we first 
contract the shared prefix of 
$\pi_{\Ghat}(t_{ij}, v_i)$
and $\pi_{\Ghat}(t_{ij}, v_j)$ into $t_{ij}$. Could we have contracted some vertex $t_{jp}$ in the process for some $p\in [d]-\{i, j\}$?  The answer is no, because by condition (2)
$V(\pi_{\Ghat}(t_{ij}, v_i)) \cap
V(\pi_{\Ghat}(t_{jp}, v_j)) \subseteq \{t_{ij}\}$. 

Thus we can proceed by contracting non-endpoints of whatever remains of $\Bunch(v_i)$ into $v_i$ as before, and get a graph $H'$ as in \Cref{fig:minor_building_proof2}. Contracting paths of length $2$ results in $K_d$.
\end{proof}

\noindent
\textbf{Remark.} It is possible to completely avoid the isolation lemma by carefully choosing specific shortest paths $\pi(t_{ij}, v_i)$ in $G$ to enforce that $\Bunch(v_i)$ is a tree and only prefixes of paths are shared between $\pi(t_{ij}, v_i)$ and $\pi(t_{ij}, v_j)$. However, the isolation lemma conveniently gives a canonical set of shortest paths, allowing us to avoid the need of explicitly describing the construction.
\end{document}